\documentclass[conference]{IEEEtran}
\IEEEoverridecommandlockouts

\def\BibTeX{{\rm B\kern-.05em{\sc i\kern-.025em b}\kern-.08em
    T\kern-.1667em\lower.7ex\hbox{E}\kern-.125emX}}

\newif\ifnotes
\notesfalse

\usepackage[space]{grffile}
\usepackage{amsmath,amssymb,amsfonts,amsthm}

\newcommand{\ignore}[1]{}
\ifnotes
  \usepackage[draft,notes=true]{dtrt}
\fi

\usepackage{algorithm}
\usepackage{algorithmic}
\usepackage{textcomp}

\newlength\myindent
\usepackage{multicol}
\usepackage{units}
\usepackage{marvosym}
\usepackage{ifsym}
\usepackage{calc} 
\usepackage[pdftex]{graphicx}
\usepackage{subfigure}
\usepackage{textcomp}
\usepackage{boxedminipage}
\usepackage{listings}
\usepackage{framed}
\usepackage{xspace}
\usepackage{bm}
\usepackage{soul,stackrel}
\usepackage{enumitem,amsmath}
\usepackage{wasysym}
\usepackage{hyperref}
\usepackage{xcolor}

\DeclareMathAlphabet{\mathcal}{OMS}{cmsy}{m}{n}

\setlist[itemize]{noitemsep,leftmargin=2mm}
\setlist[itemize,1]{leftmargin=-1mm}

\newif\iftr
\trtrue 

\newcommand{\msf}[1]{\ensuremath{{\mathsf {#1}}}}
\newcommand{\mtt}[1]{\ensuremath{\mathtt {#1}}}

\newcommand{\hash}{\ensuremath{{\mathcal H}}}
\newcommand{\adv}{\ensuremath{{\mathcal A}}}

\newcommand{\A}{\adv}
\newcommand{\samples}{\overset{\$}{\leftarrow}}
\newcommand{\functionality}[1]{\ensuremath{{\cal F}_{\textnormal{\msf{{#1}}}}}}
\newcommand{\F}{\functionality}
\newcommand{\Z}{\ensuremath{{\mathcal Z}}}

\renewcommand{\S}{\ensuremath{\mathcal S}}

\renewcommand{\L}{\mtt{L}}
\newcommand{\R}{\mtt{R}}

\newcommand{\pe}{{~\mathrel{+}\mathrel{\mkern-2mu}=~}}
\newcommand{\me}{{~\mathrel{-}\mathrel{\mkern-2mu}=~}}

\newcommand{\Fchain}{\F{Linked}}
\newcommand{\Pichain}{{\Pi_\msf{Linked}}}
\def\hsquad{\,\,\,\,\,}

\ifdefined \theorem
\relax
\else
\newtheorem{theorem}{Theorem}

\newtheorem{lemma}[theorem]{Lemma}  
\newtheorem{proposition}[theorem]{Proposition}

\fi

\newcommand{\ee}[1]{\ensuremath{#1}}

\def\deltaoffchain{\ee{\Delta}}

\def\deltareceive{\ee{\Delta}}

\def\vri{\ee{v_{r,i}}}
\def\vrip{\ee{v_{r,i}'}}
\def\vrj{\ee{v_{r,j}}}
\def\vrjp{\ee{v_{r,j}'}}
\def\sri{\ee{\sigma_{r,i}}}
\def\srj{\ee{\sigma_{r,j}}}
\def\auxin{\ee{\msf{aux}_{in}}}
\def\auxout{\ee{\msf{aux}_{out}}}
\def\inr{\ee{\msf{in}_r}}
\def\inrp{\ee{\msf{in}_r'}}
\def\sim{\ee{\mathcal{S}}}
\begin{document}

\title{Sprites and State Channels: Payment Networks that Go Faster than Lightning}

\date{}

\author{
{\rm Andrew Miller}\\
       {\small{UIUC}}
\and
{\rm Iddo Bentov}\\
{\small{Cornell University}}
\and
{\rm Ranjit Kumaresan}\\
{\small{Microsoft Research}}
\and
{\rm Christopher Cordi}\\
{\small{Sandia National Laboratories}}
\and
{\rm Patrick McCorry}\\
{\small{Newcastle University}}
}

\maketitle

\begin{abstract}
Bitcoin, Ethereum and other blockchain-based cryptocurrencies, as deployed today, cannot scale for wide-spread use.
A leading approach for cryptocurrency scaling is a smart contract mechanism called a payment channel which enables two mutually distrustful parties to transact efficiently (and only requires a single transaction in the blockchain to set-up). 
Payment channels can be linked together to form a payment network, such that payments between any two parties can (usually) be routed through the network along a path that connects them.
Crucially, both parties can transact without trusting hops along the route.

In this paper, we propose a novel variant of payment channels, called Sprites, that reduces the worst-case ``collateral cost'' that each hop along the route may incur.
The benefits of Sprites are two-fold.
1) In Lightning Network, a payment across a path of $\ell$ channels requires locking up collateral for $\Theta(\ell \Delta)$ time, where $\Delta$ is the time to commit an on-chain transaction. Sprites reduces this cost to $\Theta(\ell + \Delta)$. 2) Unlike prior work, Sprites supports partial withdrawals and deposits, during which the channel can continue to operate without interruption.

In evaluating Sprites we make several additional contributions.
First, our simulation-based security model is the first formalism to model timing guarantees in payment channels.
Our construction is also modular, making use of a generic abstraction from folklore, called the ``state channel,'' which we are the first to formalize.
We also provide a simulation framework for payment network protocols, which we use to confirm that the Sprites construction mitigates against throughput-reducing attacks.
\end{abstract}


\section{Introduction}

Blockchain-based 
cryptocurrencies such as Bitcoin, Ethereum, and others, partially derive their security from their wide replication, which unfortunately comes at the expense of limited scalability.
A leading proposal for improved scaling of cryptocurrencies is to form a network of ``off-chain'' rapid payment channels, which act like credit lines secured by ``on-chain'' currency.
In this vision for the future, today's cryptocurrencies will serve primarily as a settlement layer, such that interaction with the blockchain will not be needed for most payments.

A chief concern for the feasibility of payment channel networks is if enough collateral will be available for payments to be routed at high throughput.
For every pending payment, some money in the channel must be reserved and held aside as collateral for the duration until the payment is completed, called the ``locktime''. Even though off-chain payments complete quickly in the typical case, if parties fail (or act to maliciously impose a delay), the collateral can be locked up for longer, until a dispute handler can be activated on-chain. 
Channels can also become unbalanced or depleted if too many payments are made in one direction, requiring on-chain transactions to rebalance. If a channel is depleted, it cannot be used in a payment path. A payment fails if no paths with sufficient capacity are found.

We characterize the performance of a payment channel protocol as its ``collateral cost'', which we think of as the lost opportunity value of money held in reserve (i.e., in unit of money $\times$ time) during the locktime.
To provide wide connectivity without requiring too much collateral, payment channel networks rely on linked payments that span a path of multiple channels.
The longer the payment path, the more collateral must be reserved: for a payment of size $\$X$ across a path of $\ell$ channels, a total of $O(\ell\$X)$ money must be reserved.
The worst-case delay until the collateral is released depends on a timeout parameter.
Due to limitations in current
state-of-the-art payment channels, namely
Lightning~\cite{lightning} and Raiden~\cite{raiden}, each link in
a payment path adds an delay to the timeout parameter. This
additional delay depends on the worst-case
confirmation time for an on-chain broadcast,
which we denote by $\Delta$ (i.e., an on-chain
transaction may take $\Delta$ times longer than an
ordinary off-chain message from one party to another). Thus the worst-case
delay for a payment in Lightning is $\Theta(\ell\Delta)$,
and so the total collateral cost
(money $\times$ time) of a
$\$X$ payment over a path of length $\ell$ is
$\Theta(\ell^2 \$X\Delta)$.

\begin{figure}[t]
\includegraphics[width=\columnwidth]{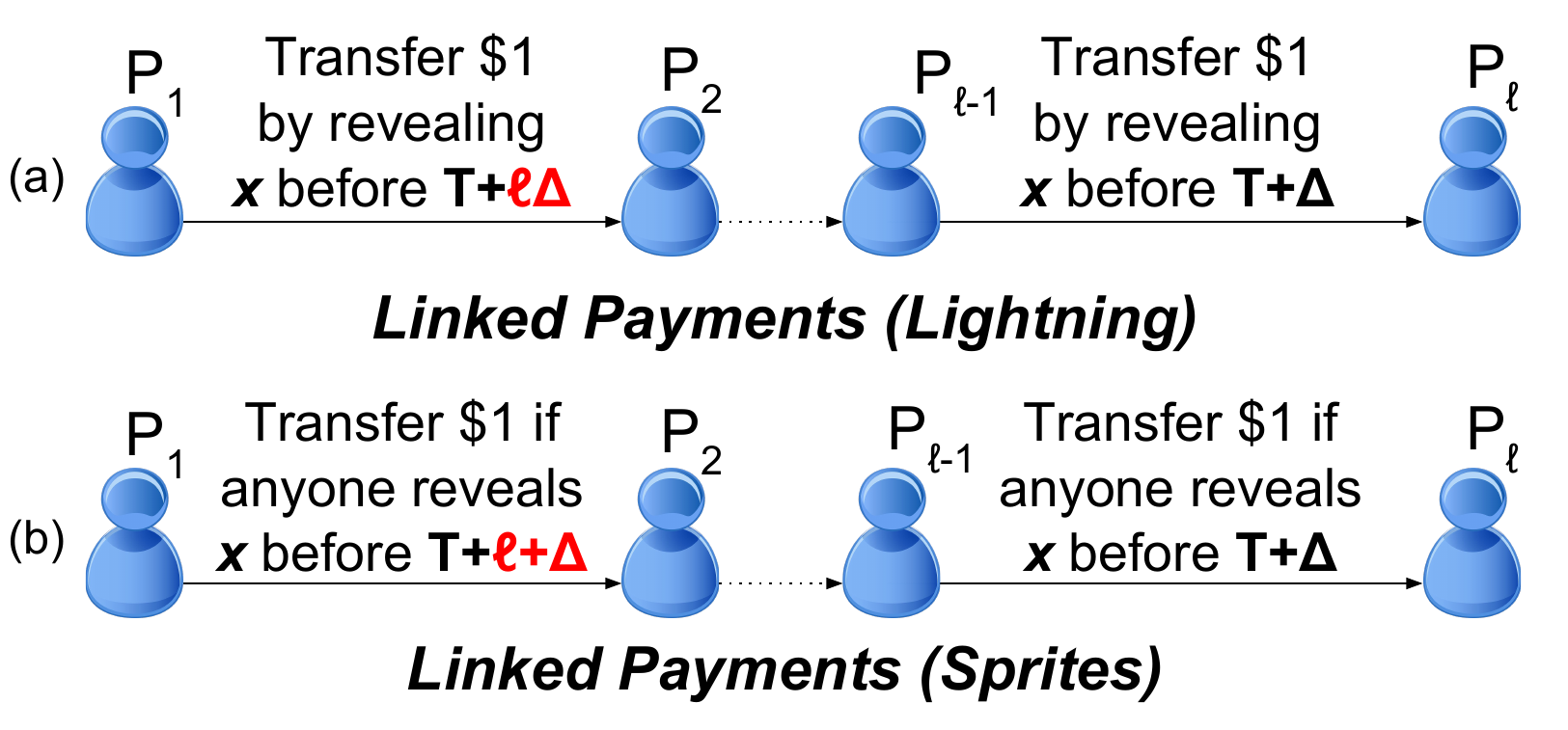}
\caption{ Cryptocurrencies like Bitcoin and Ethereum can serve as collateral for a scalable payment network (i.e. a credit network without counterparty risk)~\cite{lightning,duplexchannels}.
  Payment channels let one party rapidly pay
  another using available collateral, requiring a
  blockchain transaction only in case of
  dispute. Payments across multiple channels (a)
  can be linked using a common condition (such as
  revealing the preimage $x$ of a hash $h$). We
  contribute a novel payment channel (b) which
  improves the worst case delay for $\ell$-hop payments from $\Theta(\ell\Delta)$ to $\Theta(\ell+\Delta)$.}
\label{fig:highlight}
\end{figure}

Sprites also features an additional improvement, the ability to make incremental deposits to and withdrawals from a channel without interrupting it use for payments. In contrast, previous payment channel constructions must be closed and reopened (pausing for any pending payments to complete) in order to support deposits or withdrawals which inevitably harms throughput.

The design of Sprites is highly modular. A key contribution of our work is the formal development of a useful general primitive from folklore called a ``state channel,'' which is of independent interest. A state channel allows two or more parties to maintain an off-chain replicated state machine that can be synchronized on demand (or in case of a dispute) with the blockchain. We demonstrate the use of state channels through our Sprites construction; the abstraction neatly encapsulates the cryptography behind payment channels, such that our high level constructions need not mention digital signatures at all. We formally define and prove the security properties of our constructions using the simulation-based UC framework.



Another important question is what resulting topology will emerge from cryptocurrency payment channel networks. In the decentralized ideal, users would establish channels with peers in their social network. However, worryingly, high collateral costs associated with long payment paths may create an economic pressure towards a more centralized structure, with most individuals forming channels with only a small number of well connected bank-like hubs.
We evaluate the impact of our contribution with a simulation experiment, which features multiple topologies, and incorporates recent designs (off-chain rebalancing~\cite{revive}, and decentralized route-finding~\cite{flare}).
Not only does Sprites improve throughput and lower collateral costs, but the effect is most significant for decentralized configurations. Our work therefore directly supports realizing the vision of the decentralized  payment channel networks.

The constant locktime feature in Sprites requires the flexibility of Ethereum-style smart contracts  --- in particular, the ability for a transaction to depend on a ``global'' event recorded in the blockchain --- and therefore cannot (we conjecture) be implemented in Bitcoin. This finding therefore adds to our understanding of the separation between Bitcoin and Ethereum smart contracts, which we hope informs the design of future systems.





\ignore{
In summary, we make the following contributions:
\begin{itemize}
\item We provide the first formal definitions and secure constructions of cryptocurrency payment channels and linked payment. Our definitional framework is of independent interest, flexible composition of ``smart contracts'' (which represent code and state in a global database) and ``ideal functionalities'' which are models private to the players.
Our constructions are modular, making use of a generic ``state channel'' primitive which is of independent interest.
\item We provide a new construction for linked payment channels, which reduces the worst-case collateral cost (time $\times$ money) from $O(\ell \Delta)$ to $O(\ell + \Delta)$ for a payment chain of length $\ell$ and $\Delta$ indicates the possible delay to confirm an ``on-chain'' transaction.
\end{itemize}
}


\section{Background and preliminaries}
\subsection{Bitcoin and Blockchains}\label{sec:bgbitcoin}
Bitcoin is the first and (currently) most successful cryptocurrency, i.e. an open peer-to-peer network that implements a digital currency. The Bitcoin currency itself is not backed by anything else, but consists only of balances stored within the records of a shared blockchain database.
The average cost of a transaction in Bitcoin is currently over $\$1$ USD. It takes 10 minutes on average for a Bitcoin block to be found, ``confirming'' the transaction, but since the mining process is random, this can often take much longer.
Users are typically advised to wait for multiple (e.g., 6) confirmations before considering a transaction finalized, since forks can occur.
It is now well-understand that Bitcoin has severe performance limits. Not only does finalizing a transaction take an hour in expectation, but the overall network is limited to a throughput of around 7 transactions per second~\cite{FCW:CDE16}. For a survey on Bitcoin and cryptocurrencies, see Bonneau et al.~\cite{researchperspectives}.

The success of cryptocurrencies has further spurred interest in blockchain technology, which broadly refers to any secure database shared among multiple distrusting entities.
 In the case of decentralized cryptocurrencies like Bitcoin, these are open to the public, and are powered by the voluntary participation of anonymous ``miners.'' Alternatively, ``permissioned blockchains,'' resemble more traditional distributed state machines, rely on a defined set of participants, typically appointed by an administrative institution (e.g., a consortium of banks). In either case, blockchains derive their resilience through broad replication, which seems to come at inherent cost.

At a high level, the blockchain abstraction (see~\cite{EC:GKL14,gledger,hawk}) is an append-only replicated data structure that ensures the following properties:
 \begin{enumerate}
 \item All parties can agree on a consistent log of committed transactions.
 \item All parties are guaranteed to be able to commit new transactions in a predictable amount of time.
 \end{enumerate}
\noindent To elaborate on the latter property, what we want is a worst-case bound on how long it takes to learn about a committed transaction, then to publish a new transaction in response, and then for that transaction also to be committed. We call this time bound an ``on-chain round'' and denote it with $\Delta$. That is, an on-chain round takes at most $\Delta$ units of time. We say one unit of time corresponds to the maximum delay needed to transmit an (off-chain) point-to-point message from any one party to another.

 Nearly every blockchain cryptocurrency (such as Bitcoin, Ethereum, etc.) features (a) some built-in digital currency that can be transferred between users via on-chain transactions, and (b) some form of scripting language (e.g., Bitcoin script, or Ethereum Virtual Machine bytecode) for writing smart contracts that direct the flow of digital currency. Throughout this paper, we use smart contract pseudocode resembling both Ethereum and the UC framework~\cite{FOCS:Canetti01}, where contracts are written as reactive processes that respond to messages or method invocations.

\subsection{Blockchain scaling}
Proposed scalability improvements fall in roughly two complementary categories. The first, ``on-chain scaling'', aims to make the blockchain itself run faster~\cite{conf/uss/Kokoris-KogiasJ16,LuuNBZGS16,cryptoeprint:2016:917,NSDI:EGSR16}. A recurring theme is that the additional performance comes from introducing stronger trust assumptions about the nodes.

The second category of scaling approaches,  which includes our work, is to develop ``off-chain protocols'' that minimize the use of the block-chain itself. Instead, parties transact primarily by exchanging off-chain messages (i.e., point-to-point messages from one party to another), and interact with the blockchain only to settle disputes or withdraw funds.

\subsection{Off-chain Payment Channels}
\label{sec:background:paymentchannels}
The first off-chain protocols were Bitcoin payment channels, due to Spilman~\cite{spilman01}. In a payment channel, Alice opens a channel to Bob by initiating an on-chain deposit transaction, binding the deposit amount to a smart contract program.
The two parties can then make an arbitrary number of rapid payments between them, simply by exchanging signed messages off-chain.
A final on-chain transaction is needed to close the channel and distribute the final balance according to the code of the payment channel smart contract. The miners or validating nodes maintaining the blockchain never have to process the off-chain payments.

Spilman's payment channels only allow for unidirectional payments (i.e., the sender and receiver roles must be fixed at channel creation). Subsequent channel constructions by Decker and Wattenhofer~\cite{duplexchannels} as well as Poon and Dryja~\cite{lightning} supported ``duplex'' payments back-and-forth from either party to the other.
Because of the limitations of Bitcoin script, these constructions are subtle and require intricate workarounds (e.g., Poon and Dryja's channels require parties to store an ever-growing list of revocation keys to defend against malicious behavior).
Simpler payment channel constructions have been developed for Ethereum, based on signatures over round numbers~\cite{raiden,sparky,instantpoker}. For simplicity, we present only this latter approach; the underlying techniques are essentially the same.
An off-chain payment channel protocol roughly comprises the following three phases:

\noindent \textbf{Channel opening.} The channel is initially opened with an on-chain deposit transaction. 
This reserves a quantity of digital currency and binds it to the smart contract program.

\noindent \textbf{Off-chain payments.}
To make an off-chain payment, the parties exchange signed messages, reflecting the updated balance. For example, the current state would be represented as a signed message $(\sigma_A,\sigma_R,i,\$A,\$B)$, where a pair of signatures $\sigma_A$ and $\sigma_B$ are valid for the message $(i,\$A,\$B)$, where $\$A$ (resp. $\$B$) is the balance of Alice (resp. Bob) at round number $i$.
Each party locally keeps track of the current balance, corresponding to the most recent signed message.

\noindent \textbf{Dispute handling.} The blockchain smart contract (bound by the deposit transaction) serves as a ``dispute handler.'' It is activated when either party suspects a failure, or wishes to close the channel and withdraw the remaining balance.
  The dispute handler remains active for a fixed time, during which either party can submit evidence (e.g., signed messages) of their last-known balance. The dispute handler accepts the evidence with the highest round number and disburses the money accordingly.

  The security guarantees, roughly, are the following:    
  
  \noindent\textbf{(Liveness):}
  Either party can initiate a withdrawal, and the withdrawal is processed within a predictable amount of time. If both parties are honest, then payments are processed very rapidly (i.e., with only off-chain messages).

  \noindent\textbf{(No counterparty risk):}
  The payment channel interface offers Bob a local estimate of his current balance (i.e., how many payments he has received). Alice, of course, knows how much she has sent. The ``no counterparty risk'' property guarantees that local views are accurate, in the sense that each party can actually withdraw (at least) the amount they expect.




%

\subsection{Linked payments and payment channel networks}
\label{sec:background:linkedpayments}
Duplex payment channels alone cannot solve the scalability problem; opening each channel requires an on-chain transaction before any payments can be made. To connect every pair of parties in the network by a direct channel would require $O(N^2)$ transactions.

Poon and Dryja~\cite{lightning} developed a method for linking payments across multiple channels, suggesting the potential to connect every pair of participants through a sparse graph.
Consider the capacity graph where an edge between two participants represents an active payment channel with some available balance. If a path with sufficient capacity can be found between Alice and Bob, then Alice can send Bob an off-chain payment off-chain.

Linked payments are based on the ``hashed timelock contract'' (HTLC) for conditional payments that relies on a single hash $h = \hash(x)$ to synchronize a payment across all channels. 
Similar conditional payment techniques are used to facilitate across-blockchain transfers~\cite{tiernolan,bitcoinbook,mccorry2017atomically} and for establishing fair multiparty computation~\cite{C:BK14,CCS:KB14}.
We denote an HTLC conditional payment from $P_1$ to $P_2$ by the following:

{\centering
$$
P_1 \xrightarrow[\hsquad {h,T}\hsquad]{\$X} P_2
$$
}

\noindent which says that a payment of $\$X$ can be claimed by $P_2$ if the preimage of $h$ is revealed (i.e., by publishing a secret $x$ such that $h = \hash(x)$, via an on-chain transaction).
Otherwise, the conditional payment can be canceled after a deadline $T$.
Operationally, opening a conditional payment means signing a message that defines the deadline, the amount of money, and the hash of the secret $h = \hash(x)$; and finally sending the signed message to the recipient.
Conditional payments can also complete rapidly off-chain in the optimistic case: the sender signs a new message representing an \emph{un}conditional payment, with a higher round number to supercede the conditional payment.

Consider a path of parties, $P_1,...,P_\ell$, where $P_1$ is the sender, $P_\ell$ is the recipient, and the rest are intermediaries.
In a linked off-chain payment, Each node $P_i$ opens a conditional payment to $P_{i+1}$, one after another.

{\centering
\begin{equation} \tag{\text{\Lightning}}
P_1 \xrightarrow[\hsquad {h,T_1 = T_{\ell-1} + \Theta(\ell\Delta)}\hsquad]{\$X} P_2
~...~
P_{\ell-1} \xrightarrow[\hsquad {h,T_{\ell-1}}\hsquad]{\$X} P_\ell
\end{equation}
}

\noindent Note that the hash condition $h$ is the same for all channels. However, the deadlines may be different. In fact, Lightning requires that $T_1 = T_\ell + \Theta(\ell\Delta)$ as we explain shortly.
The desired security properties of linked payments are the following (in addition to those for basic channels given above):

\noindent\textbf{(Liveness):}
The entire chain of payments concludes (completes successfully or is canceled) within a bounded amount of time (measured in on-chain transaction cycles). If all parties on the path are honest (i.e., do not crash or fail), then the entire payment should complete successfully, using only off-chain messages (i.e., not depending on $\Delta$). 

\noindent\textbf{(No counterparty risk):}
A key desired property is that intermediaries (along the path from sender to receiver) should not be placed at risk of losing funds.
During the linked payments protocol, a portion of the channel balance may be ``locked'' and held in reserve, but it must returned by the conclusion of the protocol (regardless of whether the payment completes or cancels).%
\footnote{The intermediary nodes in a path can also be incentivized to participate in the route if the sender allocates an extra fee that will be shared among them.}
This guarantee must hold regardless of which parties are corrupted.
This property poses a challenge that constrains the choice of deadlines $\{T_i\}$ in Lightning. Consider the following scenario from the point of view of party $P_i$.

{\centering
$$
 ... ~  P_{i-1} \xrightarrow[\hsquad {h,T_i}\hsquad]{\$X} P_i
  \xrightarrow[\hsquad {h,T_{i+1}}\hsquad]{\$X} P_{i+1} ~ ...
$$
}

\noindent We need to ensure that if the outgoing conditional payment to $P_{i+1}$ completes, then the incoming payment from $P_{i-1}$ also completes.
In the worst case where $P_{i+1}$ attempts to introduce the maximum delay for $P_i$ (which we call the ``petty'' attacker), the party $P_i$ only learns about $x$ because $x$ is published in the blockchain at the last possible instant, at time $T_{i+1}$.
In order to complete the incoming payment, if $P_{i-1}$ is also petty then $P_i$ must publish $x$ to the blockchain by time $T_{i}$. It must therefore be the case that $T_{i} \ge T_{i+1} + \Delta$, meaning $P_i$ is given an additional grace period of time $\Delta$ (the worst-case bound on the time for one on-chain round).

We use the term ``collateral cost'' to denote the product of the amount of money $\$X$ multiplied by the locktime (i.e., from when the conditional payment is opened to the time it is completed or canceled). Since the payment can be claimed by time $T_\ell + \Theta(\ell\Delta)$ in the worst case, the overall collateral cost is $\Theta(\ell^{2}\$X \Delta)$ for each party (see Figure~\ref{fig:highlight} (a)). The main goal of our Sprites construction (Section~\ref{sec:highlevel}) is to reduce this collateral cost.


\subsection{Related Work}\label{sec:related}
\paragraph{Improvements to Payment Channels}
Payment channel networks have recently seen significant research interest, with several concurrent efforts to improve their performance and security.

Gervais et al.~\cite{revive} proposed the ``Revive'' protocol for rebalancing payment channels off-chain. We incorporate this into our simulation experiment in Section~\ref{sec:simulation}.

While in this paper we focus on the mechanism for executing linked payments, it remains an open problem how best to find routes through payment channel networks. Flare~\cite{flare} and Landmark Routing~\cite{privatepaymentchannels} are two proposed methods. We reproduce an experiment for the former~\cite{flare}, but leave the latter as future work.

Dziembowski et al.~\cite{perun} developed a mechanism for virtual payment channel overlays. This allows two parties with a path to establish a faster channel between them. This is complementary to our work, and we think both approaches could be fruitfully combined.

Green and Miers~\cite{journals/iacr/GreenM16} as well as Moreno-Sanchez et al.~\cite{heilman2016tumblebit} present hub-based off-chain payment protocols that offer privacy but cannot support linked payments more than one hop away from the hub.
Malavolta et al.~\cite{privatepaymentchannels} develop a privacy-oriented construction for linked payment channels, which is complementary to our work.

\paragraph{Credit networks}
Malavolta et al.~\cite{silentwhispers} developed a protocol for privacy-preserving credit networks. The main difference between a payment channel and a credit line is that payment channel balances are fully backed by on-chain deposits, and can be settled without any counterparty risk; lines of credit seem inherently to expose counterparty risk.


\paragraph{Probabilistic micropayments}
An alternative approach for off-chain micropayments is a lottery-based construction by Pass and
shelat~\cite{CCS:PassS15}. However, this requires either a semi-trusted third party, or else to lock up additional collateral (larger than the total amount of money that can be paid) to avoid a ``front-running'' attack.
Chiesa et al.~\cite{Chiesa17} also design lottery-based micropayments that provide strong privacy, but assume rational adversaries.

\paragraph{Federated sidechains}
A related proposal is to run a ``sidechain,'' consisting of an off-chain consensus protocol run amongst a set of nodes called ``functionaries,'' which jointly control a balance of on-chain deposits~\cite{pegging,strongfederations}. For example, to withdraw from a sidechain requires signatures from a majority (e.g., 5 out of 7) of the functionaries. This can be instantiated with consensus protocols that bootstrap from an existing blockchain~\cite{conf/uss/Kokoris-KogiasJ16,LuuNBZGS16,cryptoeprint:2016:917}. The main difference is that payment channels protocols guarantee security in a stronger sense, even if all the parties are corrupted, whereas in federated sidechains, a majority of functionaries misbehaving could compromise security, e.g. steal funds.

\section{Overview of the Sprites construction}
\label{sec:highlevel}

We first give a high-level overview of our construction, focusing on the main improvements versus Lightning~\cite{lightning}: constant locktimes and incremental withdrawals/deposits.
We assume as a starting point the duplex payment channel construction described earlier in Section~\ref{sec:background:paymentchannels} and presented in related works~\cite{instantpoker,raiden,sparky}).

\subsection{Constant locktime linked payments.}
To support linked payments across multiple payment channels, we use a novel variation of the standard ``hashed timelock contract'' technique~\cite{C:BK14,CCS:KB14,tiernolan,lightning}.

We start by defining a simple smart contract, called the PreimageManager ($\msf{PM}$), which simply records assertions of the form ``the preimage $x$ of hash $h = \hash(x)$ was published on the blockchain before time $T_{\msf{Expiry}}$.'' This can be implemented in Ethereum as a smart contract with two methods, $\mtt{publish}$ and $\mtt{published}$ (see Figure~\ref{fig:prot-chain-basic}).

Next we extend the duplex payment channel construction with a conditional payment feature, which can be linked across a path of channels as shown:

{\centering
\begin{equation} \tag{$\star$}
P_1 \xrightarrow[\hsquad {\msf{PM}\left[ h,T_\msf{Expiry}\right]}\hsquad]{\$X} P_2
~...~
P_{\ell-1} \xrightarrow[\hsquad {\msf{PM}\left[h,T_{\msf{Expiry}}\right]}\hsquad]{\$X} P_\ell
\end{equation}
}

\noindent In the above, the conditional payment of $\$X$ from $P_1$ to $P_2$ can be completed by a command from $P_1$, canceled by a command from $P_2$, or in case of dispute, will complete if and only if the $\msf{PM}$ contract receives the value $x$ prior to $T_\msf{Expiry}$.
As with the existing linked payments constructions~\cite{raiden,sparky}, operationally this means extending the structure of the signed messages (i.e., the off-chain state) to include a hash $h$, a deadline $T_\msf{Expiry}$, and an amount $\$X$. To execute the linked payment, each party first opens a conditional payment with the party to their right, each with the same conditional hash. Note that here the deadline $T_\msf{Expiry}$ is also a common value across all channels.

The difference between Sprites and Lightning is how Sprites handles disputes. Instead of each payment channel smart contract making a local decision about whether the preimage $x$ was revealed on time, in Sprites we delegate this to the global $\msf{PM}$ contract.
In short, each Sprites contract defines a dispute handler that queries $\msf{PM}$ to check if $x$ was revealed on time, guaranteeing that all channels (if disputed on-chain) will settle in a consistent way (either all completed or all canceled). It then suffices to use a single common expiry time $T_\msf{Expiry}$, as indicated above ($\star$).

The preimage $x$ is initially known to the recipient; after the final conditional payment to the recipient is opened, the recipient publishes $x$, and each party completes their outgoing payment. Optimistically, (i.e., if no parties fail), the process finishes after only $\ell+1$ off-chain rounds. Otherwise, in the worst case, any honest parties that completed their outgoing payment submit $x$ to the $\msf{PM}$ contract, guaranteeing that their incoming payment will complete, and thus conserving their net balance. This procedure ensures that each party's collateral is locked for a maximum of $O(\ell + \Delta)$ rounds.
%

The worst-case delay scenarios for both Lightning and Sprites are illustrated in Figure~\ref{fig:timelines}. The worst-case delay in either case occurs when an attacker publishes the preimage $x$ on-chain at the latest possible time. However, the use of a global synchronizing gadget, the $\msf{PM}$ contract, ensures that all payments along the path are settled consistently.
In contrast, Lightning~\cite{lightning} (and other prior payment channel networks ~\cite{raiden,duplexchannels,privatepaymentchannels,perun})
require the preimage to be submitted to \emph{each} payment channel contract separately, leading to longer locktimes.

\begin{figure}
\includegraphics[width=\columnwidth]{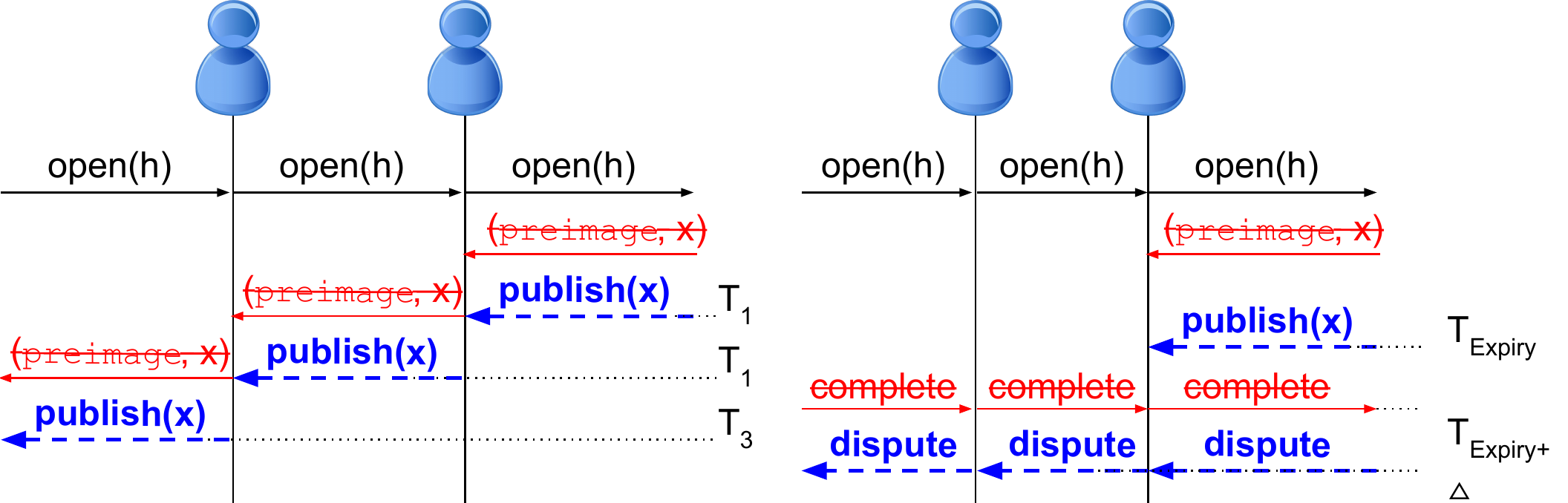}
\caption{The worst-case delay scenario, in Lightning (left) and in Sprites (right). The two parties shown are ``petty,'' dropping off-chain messages ({\color{red} \st{striken red}}) after the initial \msf{open}, and sending on-chain transactions ({\color{blue} blue}) only at the last minute. Disputes in Lightning may cascade, whereas in Sprite they are handled simultaneously.}
\label{fig:timelines}
\end{figure}

\subsection{Supporting incremental deposits and withdrawals.}
A Lightning channel must be closed and re-opened in order for either party to withdraw or deposit currency.
Furthermore, all pending conditions  must be settled on-chain and no new off-chain transactions can occur for an on-chain round ($O(\Delta)$ time) until a new channel is opened on the blockchain. 
On the other hand, Sprites permits either party to deposit/withdraw a portion of currency without interrupting the channel.

To support incremental deposits, we extend the off-chain state to include local views,  $\msf{deposits}_{\{\L,\R\}}$, which reflect the total amount of deposits from each party recorded by the smart contract. If one party proposes a view that is too stale (i.e., more than some bound $O(\Delta)$ behind), then the other party initiates an on-chain dispute. Of course, the on-chain dispute handler can read the current on-chain state directly.

To support incremental withdrawals, we implement the following. We extend the off-chain state with an optional withdrawal value $\msf{wd}_i$, which can be set whenever either party wishes to make a withdrawal (of course, both parties only sign off on such state if there is sufficient balance). The on-chain smart contract is then extended with an \mtt{update} method that either party can invoke to submit a signed message with a withdrawal value. Rather than close, the smart contract verifies the signatures, disburses the withdrawal, and advances the round number to prevent replay attacks. Further off-chain payments can continue, even while waiting for the blockchain to confirm the withdrawal.
 





\section{The State Channel Abstraction}
\label{sec:statechannel}
In this section we present the general purpose ``state channel'' abstraction, which is the key to our modular construction of Sprites payment channels.
A state channel generalizes the ``off-chain payment channels'' mechanism as described in Section~\ref{sec:background:paymentchannels}.
The state channel primitive exposes a simple interface:  a consistent replicated state machine, shared between two or more parties.
The state machine evolves according to an arbitrary, application-defined transition function. It proceeds in rounds, during each of which inputs are accepted from every party.
This primitive neatly abstracts away the on-chain dispute handling behavior and the use of off-chain signed messages in the optimistic case.

Each time the parties provide input to the state channel, they exchange signed messages on the newly updated state, along with an increasing round number.

If at any time a party fails (or responds with invalid data), remaining parties can raise a dispute by submitting the most recent agreed-upon state to the blockchain, along with inputs for the next round.
Once activated, the dispute handler proceeds in two phases.
First, the dispute handler waits for one on-chain round, during which any party can submit their evidence (i.e., the most recently signed message confirming an agreed-upon state).
The dispute handler checks the signatures on the submitted evidence, and ultimately commits the state with the highest round number. After committing the previous state, the dispute handler then allows parties to submit new inputs for the next round.

The use of the term ``state channel'' to denote a generalized payment channel appears in folklore~\cite{folklore}, however the concept has not yet been precisely formalized.
A novel feature of our model is a general way to express side effects that the state channel has on the blockchain.
Besides the inputs provided by parties, the application-specific transition function can also depend on auxiliary input from an external contract $C$ on the blockchain (which, for example, can collect currency deposits submitted by either party). The transition function can also define an auxiliary output for each transition, which is translated to a method invocation on the external smart contract $C$ (e.g., triggering a disbursement of \textbf{coins}). This feature generalizes the handling of withdrawals as transfers of on-chain currency.
We now present a security definition for state channels as an ideal functionality, followed by our construction in more detail.


%


%

%


\subsection{Modeling a State Channel as an ideal functionality}
Following several prior works,~\cite{gledger,hawk,amiller-thesis,C:BK14,CCS:KB14,CCS:KB16,CCS:KVV16,EC:KZZ15}, we formally specify our smart contract protocol as an ideal functionality, based on the UC simulation-based security framework~\cite{FOCS:Canetti01}.
%
%
The ideal functionality for state channels, $\F{State}$, is defined in Figure~\ref{fig:f:statechannel}. This functionality is parameterized by an update function, $U$, which can be customized by a developer to specialize the state channel for different applications.
The functionality proceeds in rounds,
where in each round inputs from all parties $\{P_i\}$  are accumulated
within a bounded time $O(\Delta)$. For any parties that fail
to provide input in time, a default value $\bot$ is assumed.
Finally, the state channel applies the state transaction function
to the previous round's state using the new inputs collected before broadcasting the new state to each party. 

A key technical contribution of our formalism is a generic way to capture side effects on the blockchain, which is essential for composition of smart contract protocols. Roughly, we model a hybrid world where ideal functionalities and blockchain smart contracts  interact, i.e. an ideal functionality can read from and post messages to the blockchain and invoke smart contract methods. This notion is a natural extension of the global ledger functionality~\cite{gledger}, which models the blockchain as a shared resource accessible in both the real and ideal worlds, and the commonly-used ``coins'' model~\cite{C:BK14,CCS:KB14,CCS:KB16,CCS:KVV16,EC:KZZ15}, in which ideal functionalities and parties can both send and receive money.
The $\F{State}$ functionality is parameterized with a reference to (i.e., the address of) an external blockchain smart contract $C$, with which the functionality communicates through the \mtt{aux\_input} and \mtt{aux\_output} methods.
Incoming messages from contracts are delayed by
a time of up to $\Theta(\Delta)$, reflecting the fact that
on-chain deposits are guaranteed to be available after one on-chain round.
Outgoing messages are also delayed
by up to $\Theta(\Delta)$, reflecting that an on-chain transaction
is needed to apply the on-chain action.

To summarize, in each round, the $\F{State}$ invokes the transition
update function $U$ on inputs $\msf{state}$ (the
previous state), the inputs $\{\vri\}$
supplied by the parties, and the external contract input
$\msf{aux}_{in}$ collected from $C$.
Finally, the updated state is sent to all players
within a bounded time delay of $\Theta(\Delta)$.

\begin{figure}
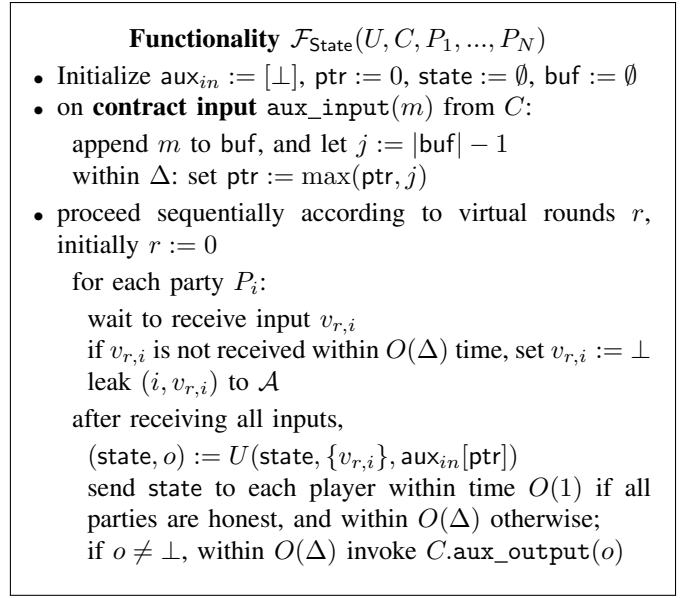

\begin{minipage}{\columnwidth}
\begin{framed}
    \centering { \bf Functionality $\F{State}(U,
      C, P_1, ..., P_N)$ } 

\begin{itemize}[leftmargin=3mm]
\item Initialize $\auxin := [\bot]$, $\msf{ptr} :=
0$, $\msf{state} := \emptyset$, $\msf{buf} := \emptyset$

\item on \textbf{contract input} $\mtt{aux\_input}(m)$ from $C$:
  \begin{itemize}
  \item[] append $m$ to $\msf{buf}$, and let $j := |\msf{buf}|-1$
  \item[] within $\Delta$: set $\msf{ptr} := \max(\msf{ptr},j)$
  \end{itemize}


\item proceed sequentially according to virtual rounds $r$, initially $r := 0$
  \begin{itemize}
  \item[] for each party $P_i$:
    \begin{itemize}
    \item[] wait to receive input $\vri$
    \item[] if $\vri$ is not received within $O(\Delta)$ time, set $\vri := \bot$
    \item[] leak $(i, \vri)$ to $\A$
    \end{itemize}
  \item[] after receiving all inputs,
    \begin{itemize}
    \item[] $(\msf{state}, o) := U(\msf{state}, \{\vri\}, \auxin[\msf{ptr}])$
    \item[] send $\msf{state}$ to each player
      within time $O(1)$ if all parties are honest, and within $O(\deltaoffchain)$ otherwise;
    \item[] if $o \ne \bot$, within $O(\Delta)$ invoke $C.\mtt{aux\_output}(o)$
    \end{itemize}
  \end{itemize}
\end{itemize}

\end{framed}
\end{minipage}
\caption{Ideal functionality for general purpose state channel}
\label{fig:f:statechannel}
\end{figure}


\paragraph{Security properties exhibited by the $\F{State}$ functionality}
As mentioned, the functionality
maintains a singular
sequential view of the current state, which is delivered consistently to each party. In each round, inputs from every party are included. The state is updated correctly according to the application-defined transition function.
We note that the 
specification provides no input privacy
(as $\F{State}$ explicitly leaks inputs to \adv),
and in fact the adversary can front-run (i.e., adversarial
inputs in a round can depend on honest party inputs).

We remark that the ideal functionality $\F{State}$ exhibits fine-grained liveness and timing guarantees.
First, in the optimistic case when both parties are honest, the payment is guaranteed to complete within a small amount of time (off-chain messages only). Even a malicious party cannot delay the advancing of rounds very much, since if they timeout their input is replaced with $\bot$ and execution proceeds.
Second, the functionality also guarantees that for each state transition with a side effect, the side effect is applied on the blockchain (i.e. the $C.\mtt{aux\_output}$ is invoked) exactly once, within a bounded time. Furthermore, the functionality guarantees that external contract inputs from $C$ (i.e., \mtt{aux\_input} method invocations) are incorporated in the inputs to the update function within a bounded time.

\subsection{Instantiating state channels}
We focus on explaining the behavior of the dispute handler smart contract, $\msf{Contract}_\msf{State}$, defined in Figure~\ref{fig:prot:statechannel:contract}, which is the protocol centerpiece; a detailed description of the local behavior for each party is deferred to the appendix (\ref{sec:prot:state:appendix}).
At a high level, the off-chain state can be advanced by having parties exchange a signed message of the following form (for the party $P_i$):
$$
\sigma_{r,i} := \msf{Sign}_{P_i}( r \| \msf{state}_r \| \msf{out}_r).
$$
where $r$ is the number of the current round, $\msf{state}_r$ is the result after applying the state transition function to every party's inputs, and $\msf{out}_r$ is the resulting blockchain output (or $\bot$ if this transition makes no output).
In the appendix we describe a leader-based broadcast protocol used to help parties optimistically agree on a vector of inputs.
We now explain how $\msf{Contract}_\msf{State}$  handles disputes, as illustrated in Figure~\ref{fig:statetransitions}.

\begin{figure}[!ht]
        { \centering \bf Protocol $\Pi_\msf{State}(U, P_1, ... P_N)$    \\}
  \vspace{2pt}
        \input{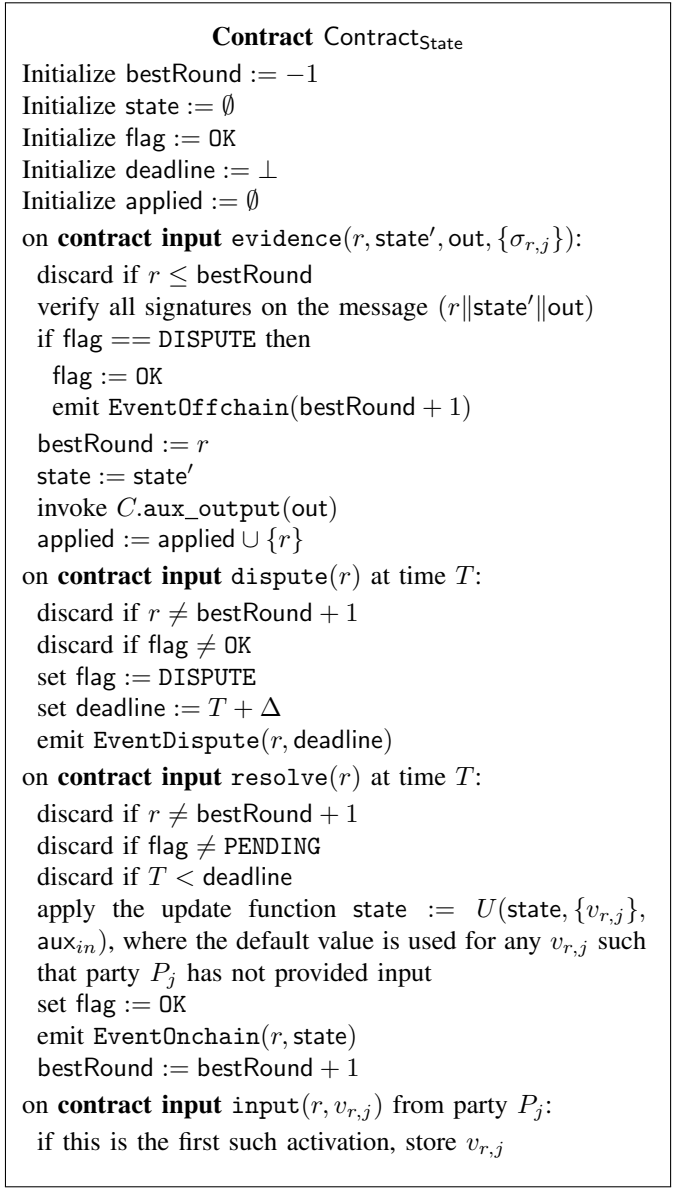}
        \caption{Contract portion of the protocol $\Pi_\msf{State}$ for implementing a general purpose state channel, $\F{State}$.}
        \label{fig:prot:statechannel:contract}
\end{figure}

\begin{figure}
  \includegraphics[width=\columnwidth]{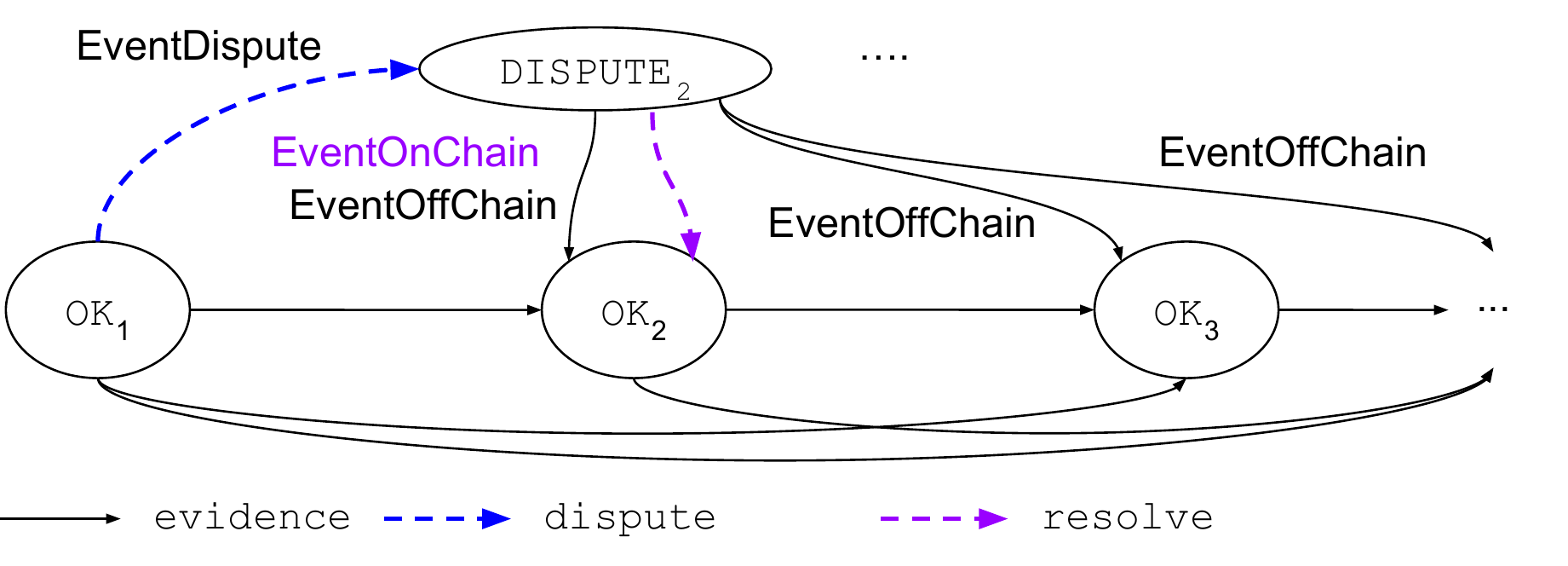}
  \caption{State transitions in $\msf{Contract}_\msf{State}$. Any party may invoke the $\msf{dispute}$ method to transition from $\mtt{OK}_r$ to $\mtt{DISPUTE}_{r+1}$. From any $\mtt{DISPUTE}_r$ state (or $\mtt{OK}_r$), an $\mtt{evidence}(r')$ message with signatures from all parties can transition to $\mtt{OK}_{r'}$ for any $r' \ge r$ or later round. If no $\mtt{evidence}$ is invoked by $T+\Delta$, the $\mtt{resolve}$ method can advance the state by applying $U$ to inputs provided on-chain.}
  \label{fig:statetransitions}
\end{figure}


\noindent \textbf{Raising a dispute.}
Suppose in round $r$ a party fails to receive signatures from all the other parties (i.e., evidence) for some $(\msf{state}_r,\msf{out}_r)$ before an $O(1)$ timeout. They then 1) invoke the $\mtt{evidence}$ method to provide evidence that round $(r-1)$ has already been agreed upon and can be used as a checkpoint, and 2) invoke the $\mtt{dispute}(r)$ method, which notifies all the other parties (\mtt{EventDispute}).

\noindent \textbf{Resolving disputes off-chain.}
Once raised, a dispute for round $r$ will be resolved in one of two ways.
First, another party may invoke the $\mtt{evidence}(r',...)$ method to provide evidence that an $r$ or a later round $r' \ge r$ has already been agreed upon off-chain, clearing the dispute (\mtt{EventOffchain}). This can occur, for example, if a corrupted node attempts to disputes an earlier already-settled round.

\noindent \textbf{Resolving disputes on-chain.}
Alternatively, if a party $P_j$ has no more recent evidence than $(r-1)$, they invoke the $\mtt{input}$ method on-chain with their input $v_{r,j}$.  After the deadline $T+\Delta$, any party can invoke the $\mtt{resolve}$ method to apply the update function to the on-chain inputs (\mtt{EventOnchain}).

\noindent \textbf{Avoiding on-chain / off-chain conflicts.}
We now explain how we avoid a subtle concurrency hazard.
Suppose in round $r$, a party receives the $\mtt{Dispute}(r,T)$ event, and shortly thereafter (say, $T+\epsilon$, for some $\epsilon>0$), receives a final signature completing the off-chain evidence for round $r$.
It would be incorrect for the party to then invoke $\mtt{evidence}(r,...)$, since
this invocation may not be confirmed until after $T+\Delta+\epsilon$.
If a malicious adversary equivocates, providing $\mtt{input}(\vrj')$ on-chain but $\vrj$ off-chain, the off-chain evidence would arrive too late.
Instead, upon receiving a $\mtt{Dispute}(r)$ event, if the party does not
already have evidence for round $r$, it pauses the off-chain routine until the dispute is resolved.




\begin{theorem}
The $\Pi_{\msf{State}}$ protocol realizes the 
$\F{State}$ functionality assuming one way functions exist. 
\end{theorem}
In the appendix we construct a simulator that translates every behavior in the real world with $\msf{Contract}_\msf{State}$ to an adversary in the ideal world with $\F{State}$, and argue that in every case the two worlds are indistinguishable.




\subsection{Modeling payment channels as an ideal functionality}
To demonstrate the use of the $\F{State}$ abstraction (and as a warmup to our full construction in Section~\ref{sec:paymentchains}) we now construct a duplex payment channel (e.g., as in~\cite{instantpoker,sparky,raiden}).
We first present our security model as an ideal functionality $\F{Pay}$.
%
Recall that a payment channel is established between two parties via a deposit of on-chain currency. Once established with a deposit of on-chain currency, the parties can rapidly pay each other by transferring a portion of this balance using only off-chain messages, resorting to interaction with the blockchain only in case of a dispute or mutual agreement to terminate. 
%
%
Our payment channel ideal functionality is defined in Figure~\ref{fig:fpay}.
It is parameterized by the (pseudonyms of) a pair of parties, $P_\mtt{L}$ and $P_\mtt{R}$, which are fixed at channel creation time (e.g., via an Ethereum transaction).
The functionality keeps track of the local balance of parties, $\msf{bal}_{\{\L,\R\}}$. It also defines a contract input method \mtt{deposit}, which can be invoked through an on-chain transaction and has the side effect of transferring \textbf{coins} from the party to the contract.
The \mtt{pay} method debits the sender's balance immediately, but increments the recipient's balance after a bounded delay. This models the fact that honest senders will immediately subtract the payment from their local view, but the recipient will only update their view after the parties reach agreement off-chain (or settle a dispute on-chain).
Finally, the \mtt{withdraw} method triggers a disbursement of \textbf{coins}.
All these method invocations are immediately leaked to the adversary, reflecting the fact that our model does not aim to capture privacy guarantees.

We now explain how the payment channel functionality $\F{Pay}$ exhibits the following properties desired of a payment channel (recalled from Section~\ref{sec:background:paymentchannels}):

\noindent\textbf{(No counterparty risk):} The functionality processes each $\mtt{withdraw}(\$X)$ message from party $P_i$ according to its record of their balances $\msf{bal}_i$, sending $\textbf{coins}(\$X)$ if $\$X \le \msf{bal}_i$.

\noindent\textbf{(Liveness):} Each payment command $\mtt{pay}(\$X)$ is processed within a bounded time. In fact the functionality provides stronger time bounds when both parties are honest, reflecting that in the optimistic case payments are completed using only off-chain communication; even in the case that some party is corrupt, progress is guaranteed within $O(\Delta)$ rounds by the on-chain dispute process.

\begin{figure}
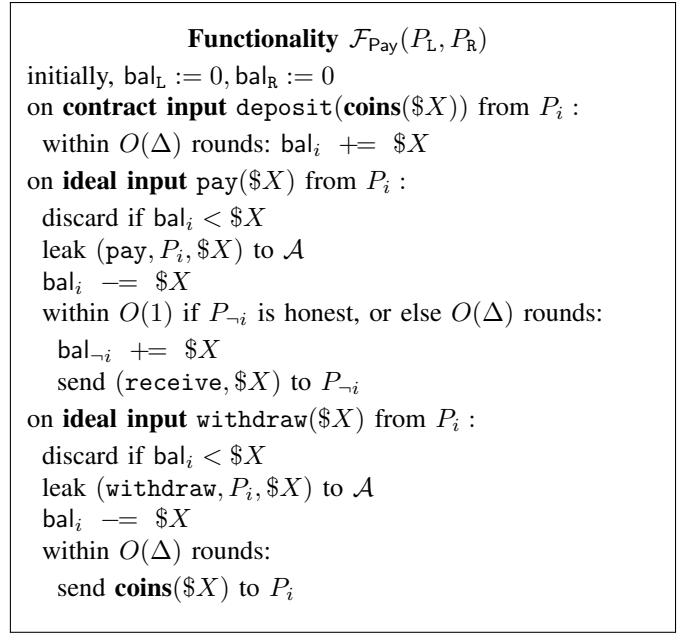

  \begin{framed}
\centering    { \bf Functionality $\F{Pay}(P_\L, P_\R)$ }

\begin{itemize}
\item[] initially, $\msf{bal}_\L := 0, \msf{bal}_\R := 0$
\item[] on \textbf{contract input} $\mtt{deposit}( \textbf{coins}(\$X) )$ from $P_i :$
  \begin{itemize}
  \item[] within $O(\Delta)$ rounds:
    $\msf{bal}_i \pe \$X$
  \end{itemize}
\item[] on \textbf{ideal input} $\mtt{pay}(\$X)$ from $P_i:$
  \begin{itemize}
  \item[] discard if $\msf{bal}_i < \$X$
  \item[] leak $(\mtt{pay},P_i,\$X)$ to $\A$
  \item[] $\msf{bal}_{i} \me \$X$
  \item[] within $O(1)$ if $P_{\neg i}$ is honest, or else $O(\Delta)$ rounds:
    \begin{itemize}
    \item[] $\msf{bal}_{\neg i} \pe \$X$
    \item[] send $(\mtt{receive}, \$X)$ to $P_{\neg i}$
    \end{itemize}
  \end{itemize}
\item[] on \textbf{ideal input} $\mtt{withdraw}( \$X )$ from $P_i:$
  \begin{itemize}
  \item[] discard if $\msf{bal}_i < \$X $
  \item[] leak $(\mtt{withdraw},P_i,\$X)$ to $\A$
  \item[] $\msf{bal}_i \me \$X $
  \item[] within $O(\Delta)$ rounds:
    \begin{itemize}
    \item[] send $\textbf{coins}(\$X)$ to $P_i$
    \end{itemize}
  \end{itemize}
\end{itemize}
\end{framed}
  \caption{Functionality model for a bidirectional off-chain payment channel.}
  \label{fig:fpay}
\end{figure}

\subsection{Constructing $\F{Pay}$ from $\F{State}$}
%

In Figure~\ref{fig:prot-pay} we give a construction that realizes $\F{Pay}$ in the $\F{State}$-hybrid world.  
Our construction consists of 1) an update function, $U_\msf{Pay}$, which defines the structure of state and the inputs provided by parties, 2) an auxiliary contract $\msf{Contract}_\msf{Pay}$ that handles deposits and withdrawals, and 3) local behavior for each party.


The update function $U_\msf{Pay}$ alone is somewhat more complicated than the $\F{Pay}$ functionality; in particular, while $\F{Pay}$ uses a single field representing the available balance of each party, $\msf{bal}_i$, the update function represents this as two fields, $\msf{cred}_i$ and $\msf{deposits}_i$.
This encoding is designed to cope with the fact
that $\F{State}$ only guarantees that auxiliary
inputs are loosely synchronized with the state
updates. If multiple deposits are received
within a short timeframe, it may be that only the most recent
deposit is passed as input to $U_\msf{Pay}$.
So when
$\msf{Contract}_\msf{Pay}$ receives a deposit of
$\textbf{coins}(x)$, we have it accumulate in a
monotonically increasing value,
$\msf{deposits}_i$, that can safely be passed to
$\mtt{aux\_input}$. The state then includes $\msf{cred}_i$ as a (possibly negative) balance offset, such that balance  available to $P_i$ is $\msf{deposits}_i + \msf{cred}_i$.
In contrast, although the $\F{State}$
functionality guarantees that each (non-$\bot$)
auxiliary output is eventually processed, they are
not necessarily in order. Since the $\F{Pay}$
functionality makes similar guarantees, it is safe
to pass the $\msf{wd}_{\{\L,\R\}}$ values directly
to $C.\mtt{aux\_output}$.
Since parties' inputs are not validated before being committed, we have $U_\msf{Pay}$  clamp each party's $\msf{pay}$ input to within the available balance, and then clamp $\msf{wd}$ to the remaining balance after that.

For the protocol to be proven secure, it must precisely match the interface of the ideal functionality. This includes exhibiting the same ``batching'' behavior. The local protocol translates $\mtt{pay}$ and $\mtt{withdraw}$ invocations into inputs of the form $(\msf{pay}_i,\msf{wd}_i)$ passed to $\F{State}$. Since $\F{State}$ accepts inputs round by round, but $\F{Pay}$ invocations may arrive at any time, the local protocol must accumulate the total of $\msf{pay}_i$ and $\msf{wd}_i$ inputs until the next $\F{State}$ round begins. However, since payments in $\F{Pay}$ are delivered one at a time, rather than batched, they must also be be delivered one at a time in the protocol. We therefore include along with the total $\msf{pay}_i$, a list $\msf{arr}_i$ of the individual payment amounts.

\begin{figure}[!ht]
  \begin{minipage}{\columnwidth}
  \begin{framed}
    \vspace{-4pt}
{\centering \bf  Update function $U_\msf{Pay}$ \\}    
\begin{itemize}
\item[] $U_\msf{Pay}( \msf{state}, (\msf{input}_\L, \msf{input}_\R), \auxin): $
  \begin{itemize}
  \item[] if $\msf{state} = \bot$, set $\msf{state} := (0, \emptyset, 0, \emptyset)$
  \item[] parse \msf{state} as $(\msf{cred}_\L, \msf{oldarr}_\L, \msf{cred}_\R, \msf{oldarr}_\R)$
  \item[] parse $\msf{aux}_\msf{in}$ as $\{\msf{deposits}_i\}_{i \in \{\L,\R\}}$
  \item[] for $i \in \{\L,\R\}$:
    \begin{itemize}
    \item[] if $\msf{input}_i = \bot$ then $\msf{input}_i := (\emptyset, 0)$
    \item[] parse each $\msf{input}_i$ as $(\msf{arr}_i, \msf{wd}_i)$
    \item[] $\msf{pay}_i := 0, \msf{newarr}_i := \emptyset$
    \item[] while $\msf{arr}_i \neq \emptyset$
    \begin{itemize}
      \item[] pop first element of $\msf{arr}_i$ into $e$
      \item[] if $e + \msf{pay}_i \leq \msf{deposits}_i + \msf{cred}_i$:
      \item[] \hspace{2mm} append $e$ to $\msf{newarr}_i$
      \item[] \hspace{2mm} $\msf{pay}_i \pe e$
    \end{itemize}
    \item[] if $\msf{wd}_i > \msf{deposits}_i + \msf{cred}_i - \msf{pay}_i$:  $\msf{wd}_i := 0$ 
    \end{itemize}
  \item[] $\msf{cred}_\L \pe  \msf{pay}_\R - \msf{pay}_\L - \msf{wd}_\L$
  \item[] $\msf{cred}_\R \pe  \msf{pay}_\L - \msf{pay}_\R - \msf{wd}_\R$
  \item[] if $\msf{wd}_\L \neq 0$ or $\msf{wd}_\R \neq 0$:
    \begin{itemize}
    \item[] $\auxout := (\msf{wd}_\L, \msf{wd}_\R)$
    \end{itemize}
  \item[] otherwise $\auxout := \bot$
  \item[] $\msf{state} := (\msf{cred}_\L, \msf{newarr}_\L, \msf{cred}_\R, \msf{newarr}_\R)$
  \item[] return $(\auxout, \msf{state})$
  \end{itemize}
\end{itemize}

\hrule
\vspace{4pt}
{\centering \bf  Auxiliary smart contract $\msf{Contract}_\msf{Pay}(P_\L, P_\R)$ \\}

\begin{itemize}
\item[] Initially, $\msf{deposits}_\L := 0, \msf{deposits}_\R := 0$
  \vspace{2pt}
\item[] on \textbf{contract input} \mtt{deposit}$(\textbf{coins}(\$X))$ from $P_i:$
  \begin{itemize}
  \item[] $\msf{deposits}_i \pe \$X$
  \item[] $\auxin.\msf{send}(\msf{deposits}_\L, \msf{deposits}_\R)$
  \end{itemize}

\item[] on \textbf{contract input} $\mtt{output}(\auxout)$:
  \begin{itemize}
  \item[] parse $\auxout$ as $(\msf{wd}_\L, \msf{wd}_\R)$
  \item[] for $i \in \{\L,\R\}$ send $\textbf{coins}(\msf{wd}_i)$ to $P_i$
  \end{itemize}
\end{itemize}

\hrule
\vspace{4pt}
{\centering \bf  Local protocol $\Pi_\msf{Pay}$ for party $P_i$ \\}
\begin{itemize}
\item[] initialize $\msf{arr}_i := \emptyset$, $\msf{pay}_i := 0$, $\msf{wd}_i := 0$, $\msf{paid}_i=0$
\item[] provide $(0,0)$ as input to $\F{State}$
\item[] on \textbf{receiving state} $(\msf{cred}_\mtt{L},\msf{new}_\mtt{L},\msf{cred}_\mtt{R},\msf{new}_\mtt{R})$ from $\F{State}$,
  \begin{itemize}
  \item[] foreach $e$ in $\msf{new}_i$:
    \begin{itemize}
    \item[] output $(\mtt{receive}, e)$
    \item[] $\msf{paid}_i \pe  e$
    \end{itemize}
  \item[] provide $(\msf{arr}_i,\msf{wd}_i)$ as input to $\F{State}$
  \item[] $\msf{arr}_i := \emptyset$
  \end{itemize}
\item[] on \textbf{input} $\mtt{pay}(\$X)$ from the environment,
  \begin{itemize}
  \item[] if $\$X \le \msf{Contract}_\msf{Pay}.\msf{deposits}_i+\msf{paid}_i{-\msf{pay}_i-\msf{wd}_i}$:
    \begin{itemize}
    \item[] append $\$X$ to $\msf{arr}_i$
    \item[] $\msf{pay}_i \pe \$X$
    \end{itemize}
  \end{itemize}
\item[] on \textbf{input} $\mtt{withdraw}(\$X)$ from the environment,
  \begin{itemize}
  \item[] if $\$X \le \msf{Contract}_\msf{Pay}.\msf{deposits}_i+\msf{paid}_i{-\msf{pay}_i-\msf{wd}_i}$:
    \begin{itemize}
    \item[] $\msf{wd}_i \pe \$X$
    \end{itemize}
  \end{itemize}
\end{itemize}
  \end{framed}
  \end{minipage}
  \caption{Implementation of $\F{Pay}$ in the $\F{State}$-hybrid world (illustrated in Figure~\ref{fig:hybridworld}).}
  \label{fig:prot-chan}\label{fig:protpay}\label{fig:prot-pay}
\end{figure}


In the appendix, we prove the following theorem:
\begin{theorem}
The $\Pi_\msf{Pay}$ protocol realizes the $\F{Pay}$ functionality in the $\F{State}$-hybrid world.
\end{theorem}


\section{Linked Payments from State Channels}
\label{sec:paymentchains}
\begin{figure*}
\subfigure[Real World (``Smart Contracts'' and\newline party-to-party communication)]{
\includegraphics[width=0.32\textwidth]{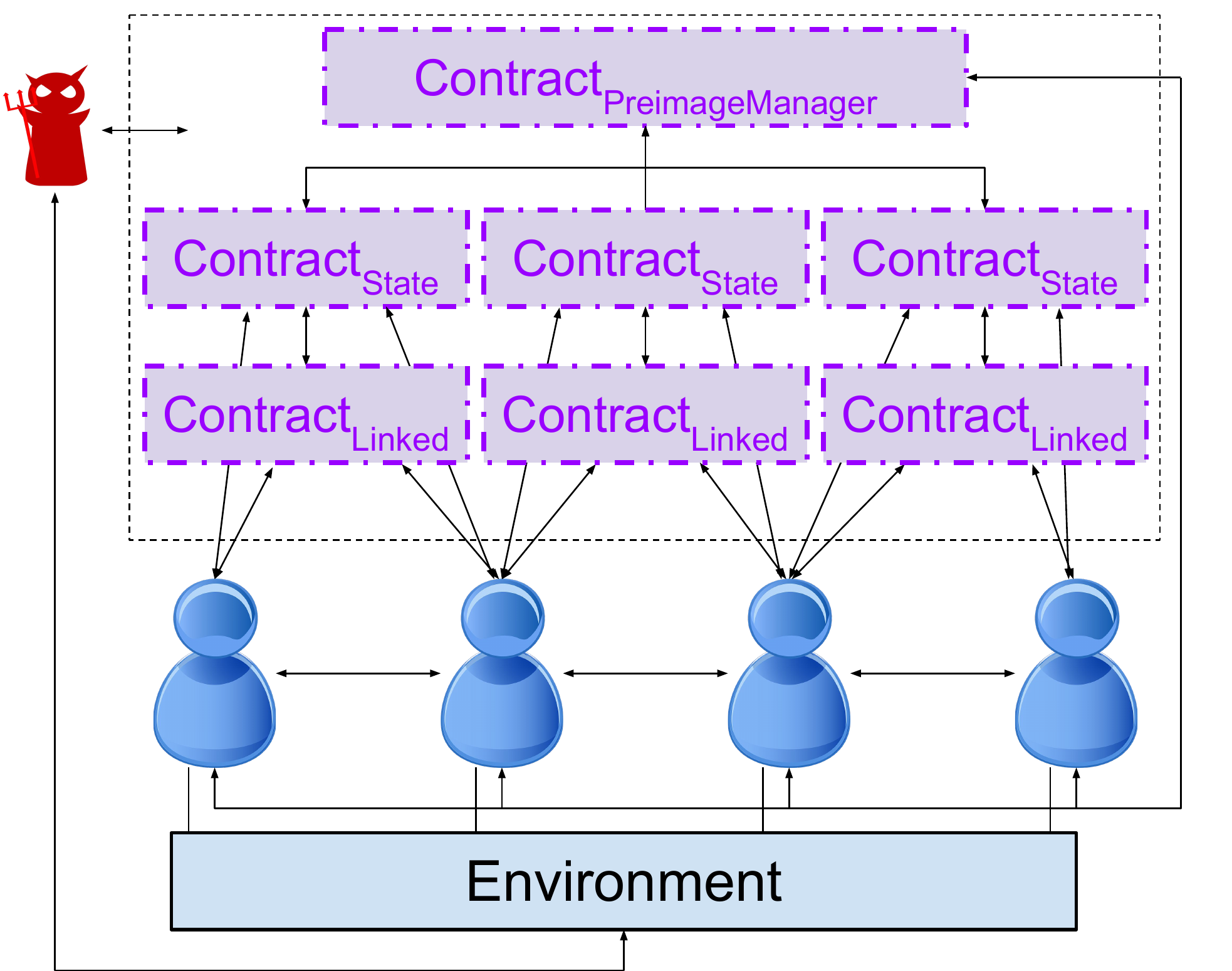}
\label{fig:realworld}
}
\subfigure[Hybrid World (Smart contracts, plus\newline the generic ``state channel'' primitive)]{
\label{fig:hybridworld}
\includegraphics[width=0.32\textwidth]{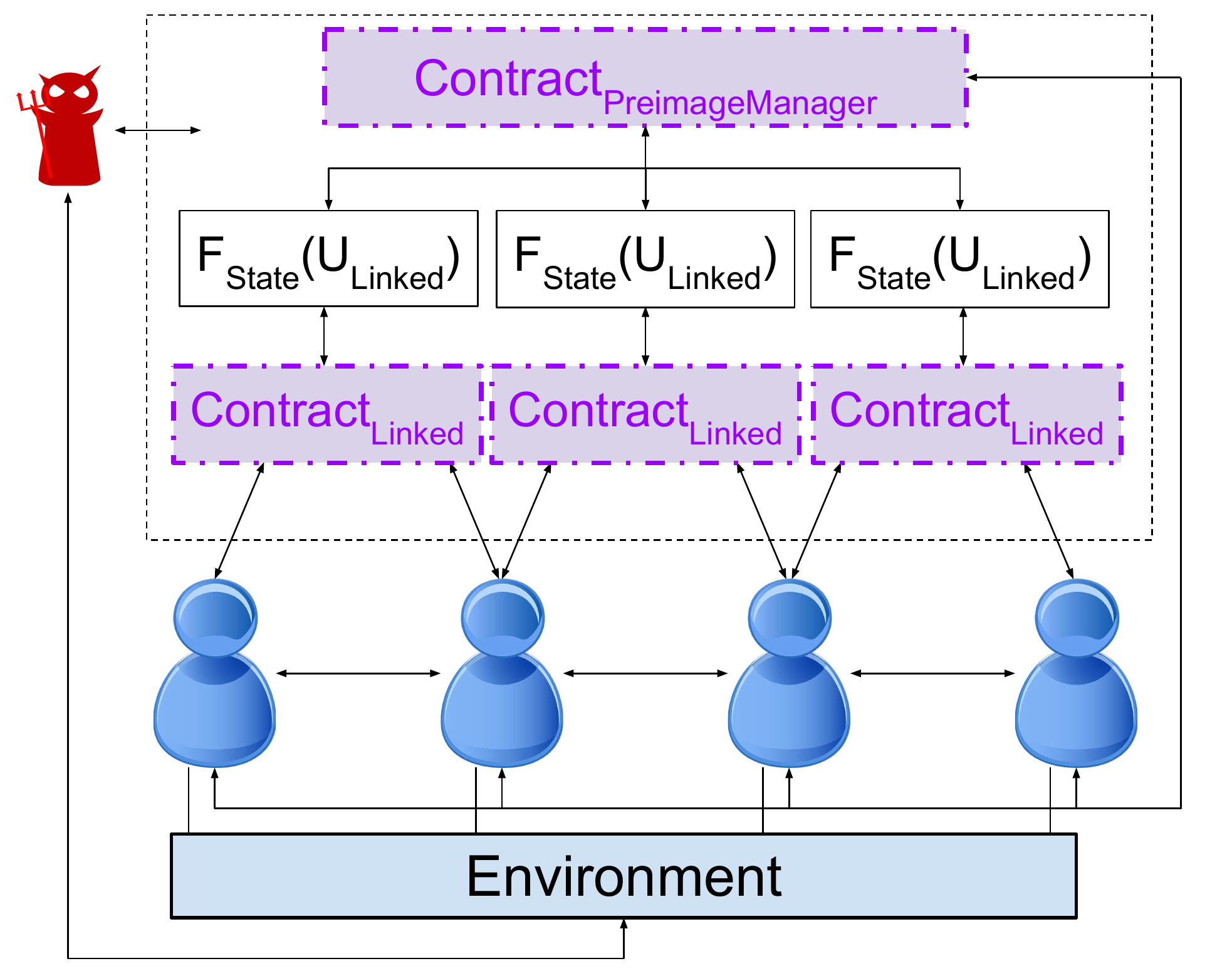}
}
\subfigure[Ideal World]{
\includegraphics[width=0.32\textwidth]{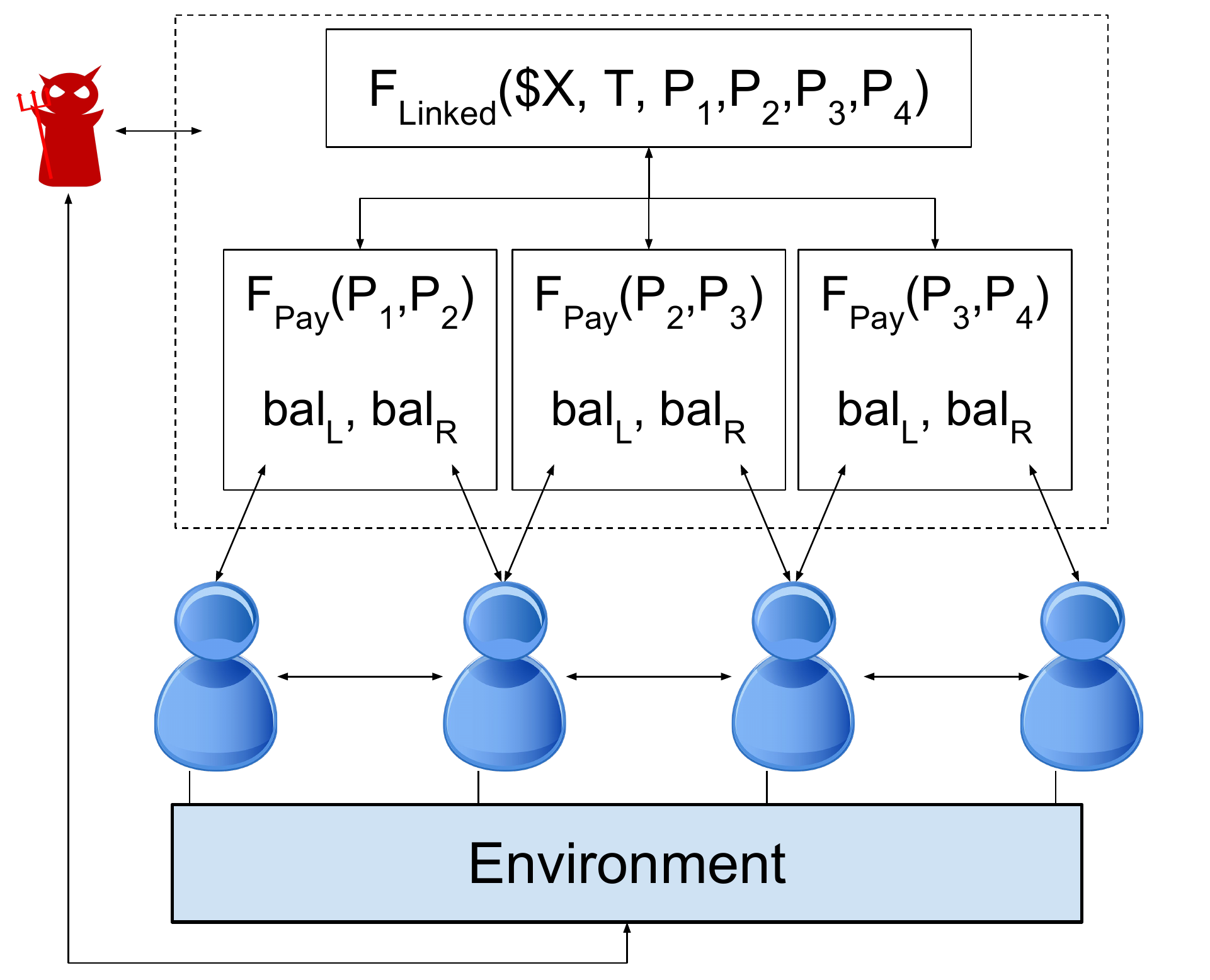}
\label{fig:idealworld}
}
\caption{Illustration of our main formalism: a construction of payment channels supporting ``atomic'' chaining (the ideal world, (c)), using a system of Ethereum smart contracts (the real world, (a)). Our construction is modular (b), using a general-purpose ``state channel'' functionality, $\F{State}$, which is of independent interest (Section~\ref{sec:statechannel}). 
%
}
\label{fig:worlds}
\end{figure*}

\begin{figure}[!ht]
\begin{framed}
  \begin{centering}
 {\bf Functionality $\Fchain(\$X, T, P_1, ... ,P_\ell)$ \\}
  \end{centering}
\begin{itemize}
\item[] initially, for each $i \in 1...{(\ell-1)}$, set $\msf{flag}_i := \mtt{init} \in \{\mtt{init}, \mtt{inflight}, \mtt{complete}, \mtt{cancel}\}$ \\

\item[] \textbf{on receiving} $(\mtt{open},i)$ from $\A$, if $\msf{flag}_i = \mtt{init}$, then
  \begin{itemize}
  \item[] if $\F{Pay}^i.\msf{bal}_\mtt{L} \ge \$X$ then:
    \begin{enumerate}
    \item[] $\F{Pay}^i.\msf{bal}_\mtt{L} \me \$X$
    \item[] set $\msf{flag}_i := \mtt{inflight}$
    \end{enumerate}
  \item[] otherwise, set $\msf{flag}_i := \mtt{cancel}$ \\
  \end{itemize}
  
\item[] \textbf{on receiving} $(\mtt{cancel},i)$ from $\A$, if at least one party is corrupt and $\msf{flag}_i \in \{ \mtt{init}, \mtt{inflight} \}$, 
  \begin{itemize}
  \item[] if $\msf{flag}_i = \mtt{inflight}$ then set $\F{Pay}^i.\msf{bal}_\mtt{L} \pe \$X$
  \item[] set $\msf{flag}_i := \mtt{cancel}$ \\
  \end{itemize}

\item[] \textbf{on receiving} $(\mtt{complete},i)$ from $\A$, if $\msf{flag}_i = \mtt{inflight}$, 
  \begin{itemize}
  \item[] $\F{Pay}^{i-1}.\msf{bal}_\mtt{R} \pe \$X$     
  \item[] set $\msf{flag}_i := \mtt{complete}$ \\
  \end{itemize}
  
\item[] \textbf{after time} $T + O(\ell + \Delta)$, or after $T + O(\ell)$ if all parties honest, raise an \mtt{Exception} if any of the following assertions fail:
  \begin{enumerate}
  \item for each $i \in 1...(\ell-1)$, $\msf{flag}_i$ must be in a terminal state, i.e., $\msf{flag}_i \in \{ \mtt{cancel}, \mtt{complete} \}$
  \item for each $i \in 1...(\ell-2)$, if $P_i$ is honest, it must not be the case that $(\msf{flag}_i, \msf{flag}_{i+1}) = (\mtt{cancel},\mtt{complete})$.
  \item if $P_1$ and $P_\ell$ are honest, then $(\msf{flag}_1,\msf{flag}_{\ell-1}) \in \{ (\mtt{complete},\mtt{complete}), (\mtt{cancel},\mtt{cancel}) \}$
  \end{enumerate}

\end{itemize}
\end{framed}
\caption{Definition of the chained-payment functionality}\label{fig:flinked}
\end{figure}

We wish to be able to route a payment from one party to another across a path of intermediary payment channels that connect them. The challenge is to ensure the collateral provided by intermediaries is returned to them within a bounded time.

\subsection{Modeling linked payment chains as an ideal functionality.}
The ideal world, which serves as our formal security definition, is illustrated in Figure~\ref{fig:flinked}(c). It essentially consists of multiple instances of the duplex channels $\F{Pay}$, as well one instance of a new linked payment functionality $\Fchain$.
Essentially, the $\Fchain$ functionality interacts with the individual $\F{Pay}$ instances, accessing their state to implement conditional payments and to ensure consistent behavior among them.

We first describe how $\Fchain$ models each conditional payment.
When a payment begins with $\mtt{open}$, $\Fchain$ reserves a portion of $P_\L$'s balance in each $\F{Pay}$ instance, advancing a status symbol $\msf{flag}$ from $\mtt{init}$ to \mtt{pending}. To reserve the balance, $\Fchain$ interacts with the individual $\F{Pay}$ functionalities directly, in a ``white box'' way, i.e. by directly manipulating the $\msf{bal}_{\{\L,\R\}}$ fields. If a channel on the path has insufficient balance, then the payment is canceled. From the \mtt{pending} state, the conditional payment must conclude (within bounded time) in one of two ways, either \mtt{cancel} in which case the balance is refunded to $P_\L$, or \mtt{complete} in which case it is paid to $P_\R$. These transitions are summarized in Figure~\ref{fig:chainstates}.

To guarantee that intermediaries do not lose money, we must ensure that if any an outgoing conditional payment completes, then the incoming payment also completes (for an honest party).
Consider a scenario where parties $P_1$ through $P_\ell$ have established $\ell-1$ payment channels, such that $\F{Pay}^i$ denotes the payment channel established between $P_i$ and $P_{i+1}$. It is easy to check that the desired properties described earlier (Section ~\ref{sec:background:linkedpayments}) are exhibited by the functionality $\Fchain$:

\noindent \textbf{(Liveness):} If all parties $P_1$ through $P_\ell$ are honest, and if sufficient balance is available in each payment channel, then the chained payment completes successfully after $O(\ell)$ rounds. More specifically, for each of channel $\F{Pay}^i$, the outgoing balance $\F{Pay}^i.\msf{bal}_\mtt{R}$ is increased by $\$x$ and each incoming balance $\F{Pay}^i.\msf{bal}_\mtt{L}$ is decreased by $\$x$.
Also note that if the sender and receiver, $P_1$ and $P_\ell$, are both honest
then the payment either completes or cancels atomically for both parties.\footnote{Note that no guarantees are provided to the sender and receiver if either misbehaves. The payment is voluntary, so the sender could simply choose not to make the payment in the first place. In future work we plan to provide a mechanism for ~\cite{FC:REFUND}.}
More precisely, after $O(\ell+\Delta)$ rounds, either the payment completes (the outgoing balance of $P_1$ is decremented by $\$X$ and the incoming balance of $P_\ell$ is incremented by $\$X$), or else the payment fails, and both parties balances remain unchanged.

\noindent \textbf{(No counterparty risk):} Even if some parties are corrupt, then the honest parties on the path, i.e. $P_{2}$ through $P_{\ell-1}$ should not lose any money. More specifically, for each party $P_i$, after a maximum of $O(\ell + \Delta)$ rounds, either the incoming balance $(\F{Pay}^{i-1}.\msf{bal}_\mtt{R})$ is incremented by $\$X$, or else the outgoing balance $(\F{Pay}^{i}.\msf{bal}_\mtt{L})$ is returned to its initial state.
First, notice that  $\msf{flag} := \mtt{cancel}$ can only occur if some parties are corrupted, or if the channel balance is insufficient. Furthermore, notice that assertion 2 ensures that honest parties do not lose money.
Finally, notice that the individual $\F{Pay}$ payment channels continue operating ``as normal'' even while the linked payment is in progress, i.e. parties can also send (unconditional) payments to each other in the meantime, as well as deposit and withdraw on-chain funds, up to the available amount.
\footnote{To generalize the functionality, we would use a generic composition operator to construct an ideal world where many instances of $\Fchain$ can exist simultaneously, and be brought into existence by parties on demand. To satisfy the composition theorem, we would need to ensure that multiple payment chains are prevented from interfering with each other, e.g. by replaying messages. We elide over these issues.}

\begin{figure}
  \includegraphics[width=\columnwidth]{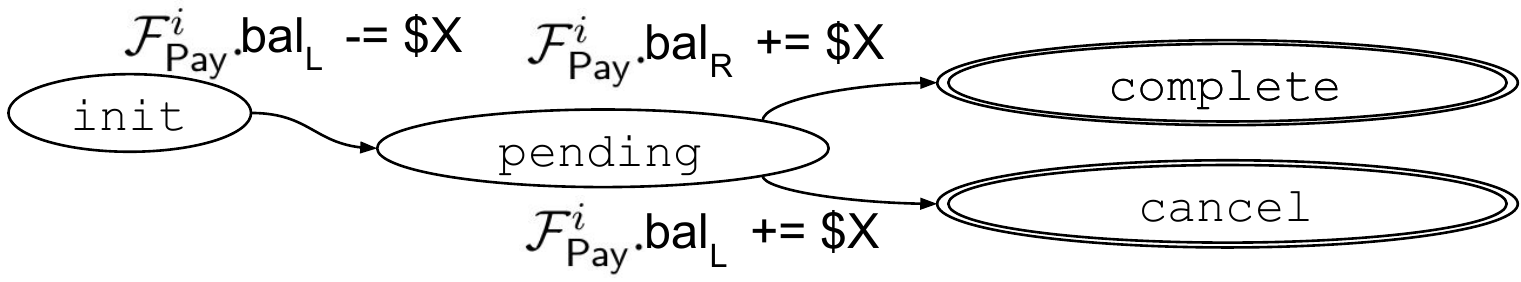}
  \caption{State transitions and side effects in $\Fchain$ (for each channel $\F{Pay}^i$. Within $O(\ell + \Delta)$ rounds, a terminal state is reached.}
  \label{fig:chainstates}
\end{figure}



\subsection{Instantiating linked payments}
As with our construction for $\F{Pay}$, our
construction $\Fchain$ consists of an update function $U_\msf{Linked}$ that
specializes $\F{State}$, as well as auxiliary
contracts and local behavior for each party (see appendix Figure~\ref{fig:prot-chain}). We
focus our discussion on the update function and
auxiliary contracts as shown in
Figure~\ref{fig:prot-chain-basic} .

The update function $U_{\msf{Linked}}$ is an outer layer around the $U_\msf{Pay}$ function (Figure~\ref{fig:prot-pay}), but extends \msf{state} to include support for a conditional payment, mirroring the status flag in the $\Fchain$ functionality. The left-hand party for each channel $P_\L$, creates a conditional payment by sending an $\mtt{open}(h)$ instruction to $\F{State}$, where $h$ is the hash of a (possibly unknown) secret.
Each conditional payment can be concluded in one of three ways: by a $\mtt{complete}$ instruction from $P_\L$, a $\mtt{cancel}$ message from $P_\R$, or through a $\mtt{dispute}$ case, which can occur only if one of the two parties fails, as we describe shortly.

To establish a chain of linked payments, the initial sender $P_1$ first creates a secret $x$, shares with the recipient $P_\ell$, and creates an outgoing conditional payment to $P_2$ using $h = \hash(x)$. Each subsequent party $P_i$ in turn, upon receiving the incoming conditional payment, establishes an outgoing conditional payment to $P_{i+1}$. Once the recipient $P_\ell$ receives the final conditional payment, it multicasts $x$ to every other party.

The key challenge is to ensure that if an honest party's outgoing conditional payment completes, then its incoming conditional payment must also complete. In the $\mtt{dispute}$ case, whether the conditional payment is canceled or refunded depends on the state of the global preimage manager, $\msf{Contract}_\msf{PM}$, which acts like a global condition: if the preimage manager contract receives $x$ before time $T_\msf{Expiry}$, then \textit{every} conditional payment that is disputed will complete; otherwise, \textit{every} disputed conditional payment will cancel. Therefore, if an honest party receives $x$ before time $T_\msf{Expiry}-\Delta$, it is safe to \mtt{complete} their outgoing conditional payment, since in the worst case they will be able submit $x$ to $\msf{Contract}_\msf{PM}$ and claim their incoming payment via \mtt{dispute}.

In the Appendix we give a security proof that $\Pichain$ realizes the ideal world $(\F{Linked},\F{Pay})$. Here we just remark on the three scenarios the protocol is designed to handle.
First, it is possible that parties receive
inconsistent values of $h$. But since each party
creates an outgoing conditional payment with $h'$
only after receiving an incoming conditional
payment with the same hash $h'$, their balance is
preserved regardless.
Second,
note that the ideal functionality permits for some
conditional payments to \mtt{complete} while
others \mtt{cancel}, but only in ways that do not
harm honest parties (only corrupted parties may lose money).
Finally, we note
that in the optimistic case, when all parties are
honest and thus all payment conditional payments
\mtt{complete}, the $\msf{Contract}_\msf{PM}$ is
never invoked at all.



\begin{figure}[!t]
  {\centering \bf  Protocol $\Pichain(\$X, T, P_1, ... P_\ell)$ \\}
        \vspace{2pt}
  \input{sections/prot_chain_contract}
  \caption{Smart contract for protocol $\Pichain$ that implements linked payments ($\Fchain,\F{Pay}$) in the $\F{State}$-hybrid world. See Appendix (Figure~\ref{fig:prot-chain}) for local behavior.}\label{fig:prot-chain-basic}
  \end{figure}



\section{Simulating Payment Channel Networks}
\label{sec:simulation}
In our Sprites construction we optimize the locktimes and collateral costs in payment channels. We hypothesize that Sprites will lead to better performance in a real system (especially if under attack) because more collateral will be unlocked and available to route payments.
To estimate the impact overall system performance, i.e., payment throughput, we developed a simulation framework to model Lightning and Sprites payment channels in various network configurations.
We generate synthetic network topologies based on two models, scale-free and small-world, similar to how Prihodko et al~\cite{flare} evaluated the Flare routing scheme; in fact,
in the Appendix
we present a reproduction of their experiment to validate our framework.
Our simulation also relates to that of Moreno-Sanchez et al.~\cite{pathshuffle,silentwhispers}, although they use data from the Ripple network rather than synthetic topologies. They also feature an alternative routing protocol which we leave for future work.

\paragraph{Topology formation}
It remains to be seen what payment channel nework topologies will emerge in practice.
In our simulation, we consider two commonly used topology models, scale-free networks and small-world networks. A small-world network is a graph where the average path length between nodes is short, at most $O(\log N)$. A scale-free network is one where the degree distribution is given by the power law~\cite{scalefree}. Scale-free networks are also small world, but in particular feature high-degree hubs.
The Barab\'asi-Albert (BA) model~\cite{bamodel} is an algorithm for generating a scale-free network, while the Watts-Strogatz (WS) model~\cite{wsmodel} is an algorithm for generating small world networks that are not scale free.
Roughly speaking, BA models a more centralized network because of the influence of large hubs.
%
For our experiment we generate random undirected graphs of 2,000 nodes using the BA and WS algorithms.
We also
assign to each node a network latency; based on estimates from measurements of the Bitcoin network conducted by Neudecker et al.~\cite{bitcointiming},%
\footnote{The latency captured by Bitcoin network measurements from Neudecker et al.~\cite{bitcointiming} includes internet latency as well as Bitcoin-specific transaction relaying behavior.}
we sample so that 92.5\% of nodes have 100ms latency, 4.9\% have 1 second latency, and 2.6\% have 10 second latency.

\paragraph{Generating payment requests}
To generate distributions of payments, we make use of an anonymized dataset of credit card transactions provided to us by a bank. The dataset consists of four million transactions, made by approximately 50,000 unique cardholders (identified only by random labels), over a six month period (from Dec 2016 to May 2017).
We assign attributes to each node by independently sampling from distributions as follows.
We label each node as either a ``consumer'' with probability $1/3$ or ``merchant'' with probability $2/3$, as there are twice as many merchants as consumers in our anonymized dataset.
We assign an initial payment channel capacity to each
edge by sampling from a bi-modal distribution, either ``High'' (\$800)
with probability 0.2 or ``Low'' (\$50) with probability 0.8.
We assign each consumer node a transaction value mean and variance by sampling from the anonymized dataset.
Finally, we assign a ``spend frequency'' (resp. ``receive
frequency'') to each consumer node (resp. merchant node) by sampling
from the anonymized dataset.
To create a payment request, we sample a consumer node at random (weighted by its spend
frequency) as the sender, and a merchant node (weighted by receive
frequency) as the recipient. For the transaction amount, we sample from a normal distribution with using the mean and variance associated with the sender node.

\begin{figure}[t!]
  \centering
  \subfigure[Scale Free (BA)]{
    \includegraphics[width=0.5\columnwidth]{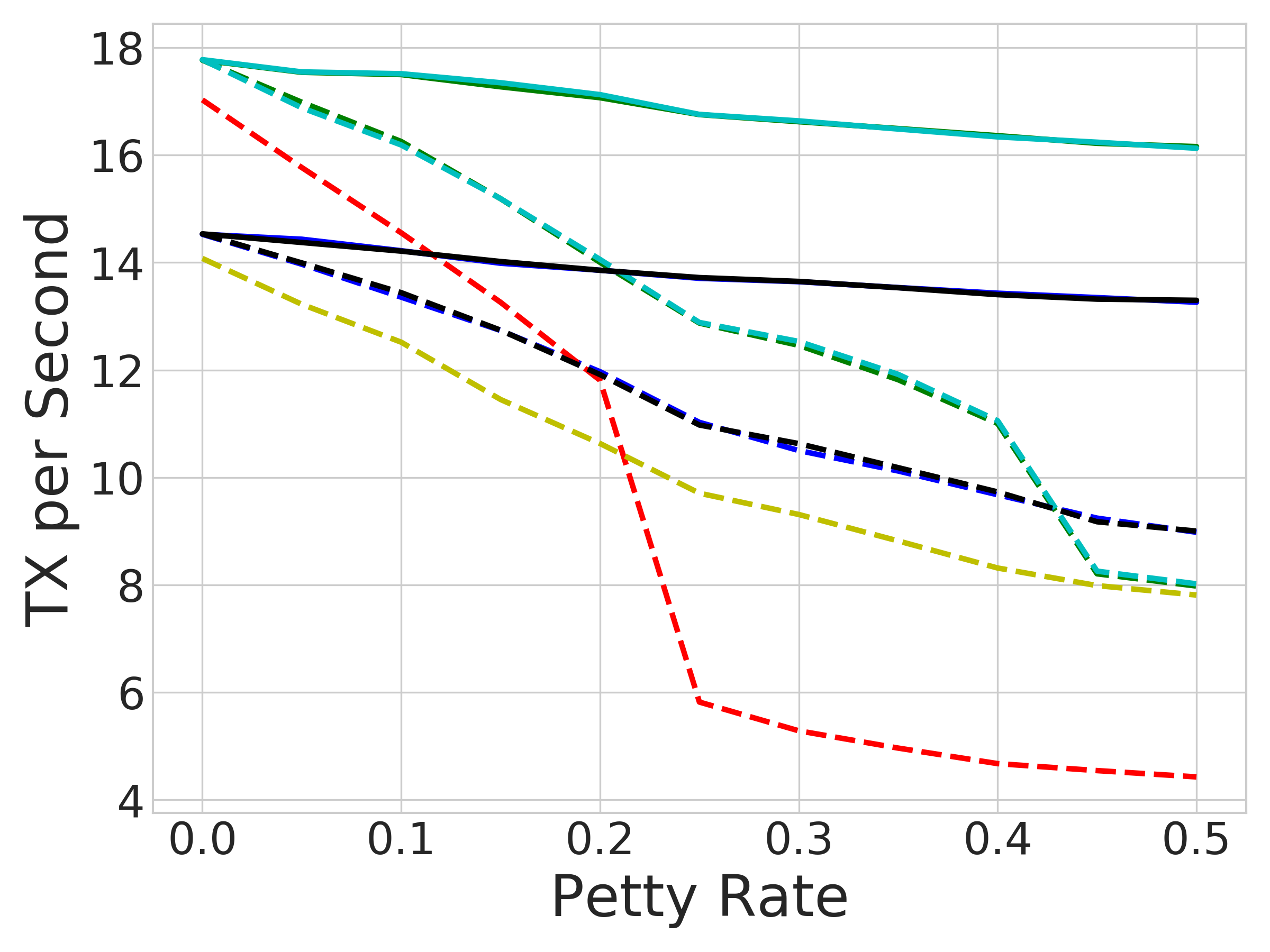}
  }%
  \subfigure[Small World (WS)]{
    \includegraphics[width=0.5\columnwidth]{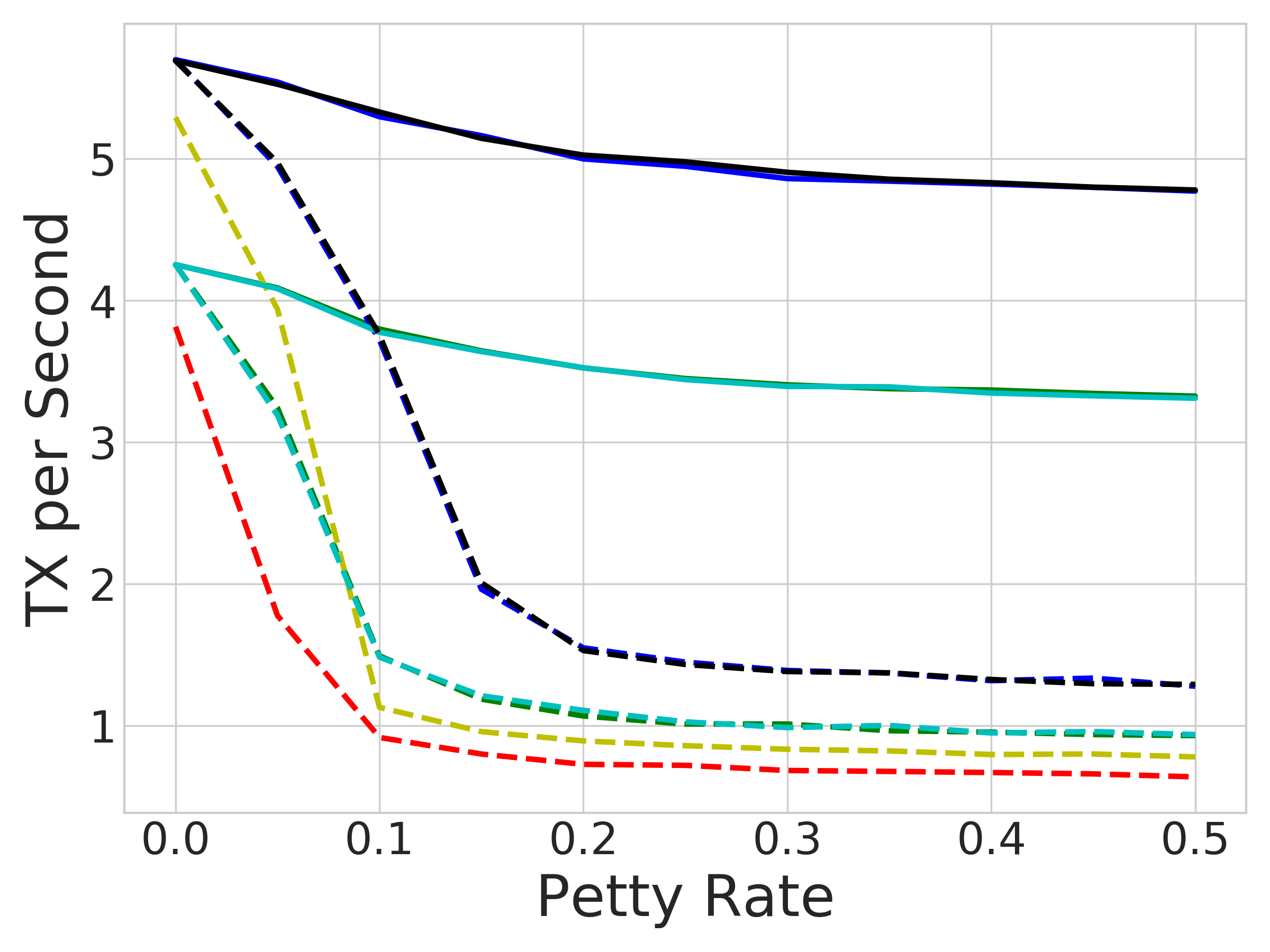}
  }
  \caption{Effect of petty attacks on transaction throughput (@98\% success rate), under varying configurations of payment channel networks.}
  \label{fig:petty-goodput}
\end{figure}
\begin{figure}[t!]
    \subfigure[Scale Free (BA)]{
      \includegraphics[width=0.5\columnwidth]{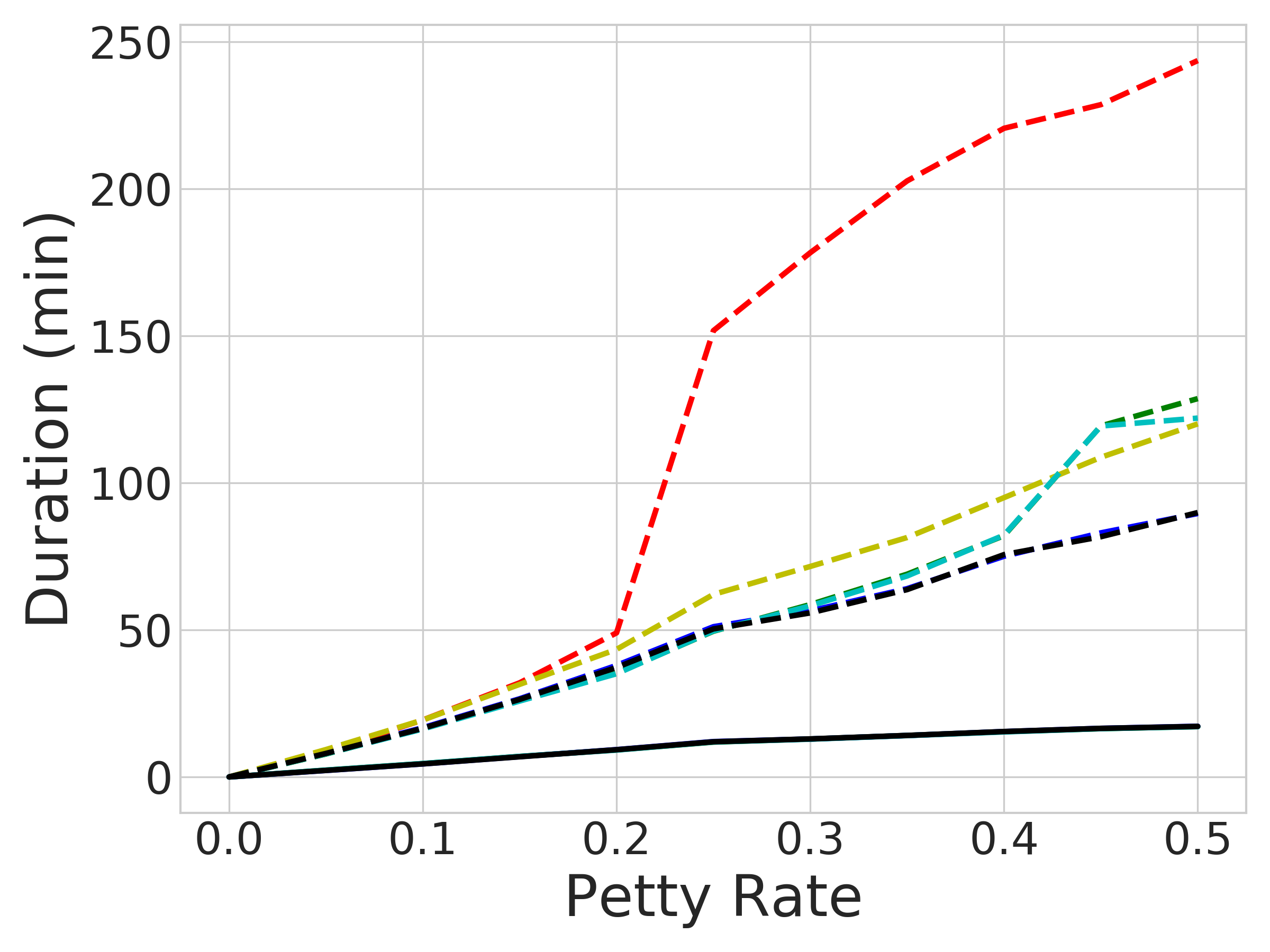}
    }%
    \subfigure[Small World (WS)]{
      \includegraphics[width=0.5\columnwidth]{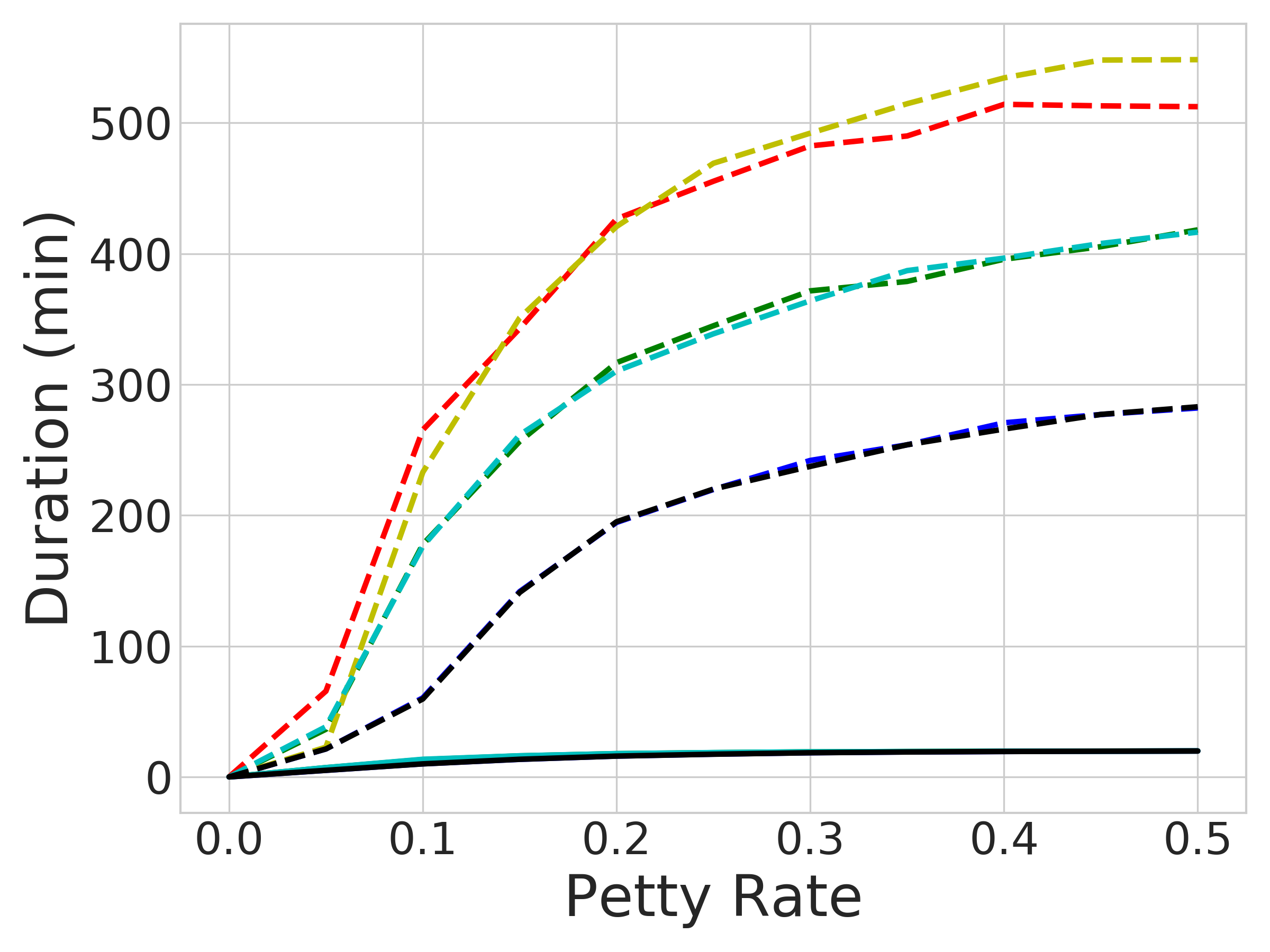}
    }
    \centering
\includegraphics[width=\columnwidth]{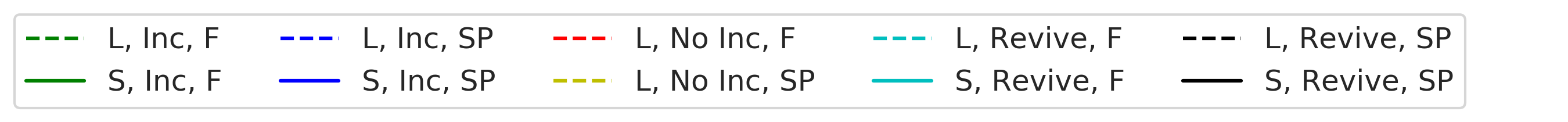}
    \caption{Effect of petty attacks on transaction duration (@98\% success rate), under varying configurations of payment channel networks.}
    \label{fig:petty-duration}
\end{figure}

\paragraph{Route finding}
For each payment request, we attempt to find a route using one of two algorithms: 1) Flare (F)~\cite{flare} is a decentralized DHT-like route finding protocol, that requires only local information and interaction between nodes.
See the appendix for more background on Flare.
2) Shortest Path (SP) models an idealized central party that uses global information to determine the overall shortest path. In future work we plan to evaluate landmark routing~\cite{silentwhispers} as well, which we hypothesize would be a middle ground between these two. If a suitable path cannot be found, we count the payment attempt as a failure.

\paragraph{Evolution}
Our simulation keeps track of the balance and pending conditions in each payment channel, updating them tick by tick (each tick represents 1 second of real time).
When a route is found, a new conditional payment is opened on each channel immediately (that is, we model the ``open'' command as completing instantaneously). Each payment channel supports potentially many concurrent in-flight payments, as long as the channel has sufficient balance. For each in-flight conditional payment, the ``complete'' commands propagate one hop at a time, where the time to complete each hop is based on the node's network latency. 

\paragraph{Rebalancing}
Since each payment originates at a consumer (a source), and terminates at a merchant (a sink),  the payment channels become unbalanced over time. We model on-chain rebalancing behaviour as follows:
at each tick (every 10 seconds) every  node checks if it has low balance channels (less than \$20,
the maximum size of one payment) and if so creates an on-chain transaction to replenish their funds evenly across their channels. In the incremental deposits (No Inc.) setting, only the amount to rebalance is locked for an $\Delta$ period, while otherwise (No Inc.) the entire channel is paused for $\Delta$. We also model Revive~\cite{revive} off-chain rebalancing (Revive). Every three ticks (30 seconds), we adjust the channel balances to minimize the difference between incoming and outgoing balances along each, effectively unwinding any credit cycles in the network. We model off-chain rebalancing as instantaneous though in reality it would incur some delay.

\paragraph{Modeling ``petty'' attacks}
The main benefit of the Sprites technique is to provide a better worst-case delay. We therefore evaluate an attack scenario where the attacker's goal is to reduce transaction throughput. We let the attacker control a varying fraction $0 \le \beta \le 0.5$ of the network. The attacker follows a ``petty'' strategy, delaying state-channel messages until the last possible moment.

\paragraph{Quantifying throughput}
We are especially interested in the maximum attainable throughput enabled by payment channel networks. To establish a normalized baseline across configurations, we first identify the maximum throughput such that at least 98\% of transaction attempts succeed; that is, we increase the payment generation rate until the simulation reaches a steady state where 2\% of payment attempts fail.

\paragraph{Results}
In Figure~\ref{fig:petty-goodput} we show the results of running our simulation for several different configurations. We 
We varied several parameters, namely incremental deposits, and constant locktime (S) vs lightning (L), all at varying levels of petty attacker. 
We then report the rate of successful transactions per second. (at 98\% success).
The $\msf{PettyRate}=0$ case is our baseline, from which we determined the request rate at which 98\% of payment attempts succeed.

In every scenario, incremental deposits improves the attainable throughput by over 3.2\% in the BA network, and over 11.5\% in the WS network. We also find that Revive appears to have minimal effect, presumably because credit cycles do not often form.

We next consider the extent to which a malicious attack could disrupt the throughput of the payment channel network.
We also see that in every case, the use of constant locktimes effectively mitigates the harm caused by petty attackers. When half the network is petty $\msf{PettyRate}=0.5$,  the Sprites model sustains 88.7\% (for BA) of its baseline throughput and 76.8\% (for WS), whereas Lightning drops to at best 60.9\% for BA and 22.7\% for WS.
We can explain this by the average duration of payment requests, as shown in Figure~\ref{fig:petty-duration}. Under increasing petty attacks, many transactions are routed through paths that include a petty attacker, leading to longer transaction times and less available collateral.
Finally, we note that in all cases, the relative improvement of Sprites vs Lightning is more pronounced in the more decentralized small world (WS) model rather than the more centralized scale free model (BA).

\section{Discussion and Conclusion}
Cryptocurrencies face several ongoing challenges: they must scale up to accommodate increasing user demand, and they must compete with centralized alternatives.
%
%
Our construction of Sprites embodies two novel insights that improve the achievable throughput and worst-case collateral costs compared to Lightning~\cite{lightning}, the current state-of-the-art design. 
Furthermore, we show through simulation experiment that the improvements in Sprites have a greater impact for decentralized topologies and routing algorithms. Our work therefore directly supports a more decentralized payment channel network.

We now outline several further questions for future work.

\paragraph{Feasibility of constant locktimes in Bitcoin}
Our constant locktimes construction relies on a global contract mechanism, which is easily expressed in Ethereum, but cannot (we conjecture) be emulated in Bitcoin without some modification to its scripting system. Are there minimal modifications to Bitcoin script that would enable constant locktimes?

\paragraph{Privacy}
As we have focused on collateral costs as our key performance objective, our constructions and security definitions do not aim to ensure transaction privacy. Our work is therefore complementary to efforts that focus primarily on privacy~\cite{journals/iacr/GreenM16,heilman2016tumblebit,privatepaymentchannels,instantpoker}. We believe our state channel abstraction can serve as a convenient building block for this important future work.

\paragraph{Concurrent Conditional Transfers}
Concurrency in payment channels has been explicitly studied by Malavolta et al.~\cite{privatepaymentchannels}.
For simplicity, our functionality model $\Fchain$ expresses only a single conditional transfer; it would be straightforward to extend this to multiple concurrent payments.
We have included this feature in our proof of concept implementation in Ethereum,
\iftr
found in the Appendix.
\fi

\paragraph{Supporting fees}

Participants who act as intermediaries in a payment path contribute their resources to provide a useful service to the sender and recipient. The intermediaries' collateral is tied up for the duration of the payment, but the sender and recipient would not be able to complete their payment otherwise. Therefore the sender may provide a fee along with the payment, which can be claimed by each intermediary upon completion of the payment. To achieve this, each conditional payment along the path should include a slightly less amount than the last; the difference can be pocketed by the intermediary upon completion. The following example provides a $\$1$ fee to each intermediary, $P_2$ and $P_3$.

{\centering
\vspace{-15pt}
$$
P_1 \xrightarrow[\hsquad {\msf{PM}[h,T_\msf{Expiry}]}\hsquad]{\$X + 2\textcent}
P_2 \xrightarrow[\hsquad {\msf{PM}[h,T_\msf{Expiry}]}\hsquad]{\$X + 1\textcent}
P_3 \xrightarrow[\hsquad {\msf{PM}[h,T_\msf{Expiry}]}\hsquad]{\$X} P_4
$$
}

\paragraph{Fair Exchange of Invoices}
McCorry et al proposed that the widely used Payment Protocol standard, BIP70, should be revised to include digitally signed invoices~\cite{FC:REFUND} from the merchant (i.e. the payment acknowledgement message).
We highlight that payment protocols such as BIP70 can be supported in Sprites and that it is feasible to fairly exchange a merchant's invoice for the customer's payment using the conditional transfers. The merchant would sign an invoice that includes the hash $h = \hash(x)$, and send this invoice to the customer. The customer creates a conditional transfer such that the merchant can receive these coins if the preimage $x$  is revealed. Finally, the merchant reveal $x$ to complete the Sprites payment.



\bibliographystyle{plain}
\emergencystretch 1.5em
\bibliography{bibstatechannel}

\appendix

\iftr



\subsection{Ideal Functionalities and Simulation Based Security}
Our formalism for Sprites is founded on the simulation-based security framework (in particular Universal Composability (UC)~\cite{FOCS:Canetti01}) which is a general purpose framework for modelling and constructing secure protocols. We now give 

\paragraph{Notes on the Blockchain Model}
We use the typical idealized model of a Bitcoin-like blockchain~\cite{hawk,CCS:KB14,CCS:KMB15,CCS:KMB15} as described below. 
For our purposes, a blockchain functions as a shared public database. Any party can write to the blockchain by submitting a ``transaction,'' which propagates throughout the network and is eventually committed into a consistent ordered log. Every party can view all the transactions committed on the blockchain; however, the views are only approximately synchronized. If one party's view comprises the sequence $\msf{txs}_1$ and another party's view comprises the sequence $\msf{txs}_2$, then it must be that $\msf{txs}_1$ is a prefix of $\msf{txs}_2$ or vice-versa. Furthermore, if any party's view includes a transaction $\msf{tx}$, then every party's view will also include $\msf{tx}$ after a maximum time bound. To simplify matters, we consider a single time bound, $\Delta$, which bounds the maximum delay ``round-trip'' time for a blockchain transaction: if some party submits a transaction $\msf{tx}$ at time $T$, then every party sees $\msf{tx}$ confirmed by time $T+\Delta$.

\paragraph{Smart contract programming conventions}
We make use of smart contract programming conventions in our construction, inspired by those implemented in Ethereum. Smart contracts are processes running in the blockchain database that accept input via user-submitted transactions. Smart contracts can be trusted to execute correctly, but do not provide any inherent privacy;  the adversary also has the opportunity to reorder and front-run user-submitted inputs.

We make use of several conventions based on features typically found in smart contract programs. Smart contracts have access to a clock (i.e. a ``block number'') which is approximately synchronized (i.e., to within $\Delta$) of the honest parties. We also assume that the smart contract execution environment provides a built-in notion of \textbf{coins}, which can be transferred (conserving total balance) between contracts and parties.
Contracts can also emit events as a way of notifying parties.
A party receives the event when the transaction triggering that event is confirmed by multiple blocks in the blockchain.

\paragraph{Universal Composability}
Universal Composability~\cite{FOCS:Canetti01} is formally defined in an execution model involving a system of interactive Turing machines (ITMs). ITMs are defined in a reactive style, by describing how to behave upon receiving a message; the resulting behavior includes modifying a local state, and sending a message to another ITM process.
The UC execution model involves several kinds of processes: an environment, $\Z$, which represents the ``external world'' and chooses the inputs given to each party and observes the outputs; parties that follow a given protocol $\Pi$, and an Byzantine adversary $\A$ that controls corrupted parties. The model also includes functionalities, $\F{}$, which act like idealized trusted third parties. A functionality serves as the target specification; the ``ideal world'' contains a functionality that exhibits all the intended properties of the protocol. A functionality in the ``real world'' is also used to represent network primitives and setup assumptions.
A proof in this framework takes the form of a simulator, which translates every attacker $\A$ in the real world into a simulated attacker $\S_\A$ in the ideal world, such that the two worlds are indistinguishable to the environment; in other words, the real world is just as good as the ideal world.
We denote by $\msf{execReal}(\Z, \Pi, \A)$ the output of $\Z$ following its interaction with $\A$ and the honest parties in an execution of $\Pi$ in the real world, and we denote by $\msf{execIdeal}(\Z, \F{}, \S_\A)$ the output of $\Z$ following an execution of $\F{}$ with $\S_\A$ and the honest parties in the ideal world.
Thus, we say that protocol $\Pi$ in the real world realizes functionality $\F{}$ in the ideal world if the distributions $\msf{execReal}(\Z, \Pi, \A)$ and $\msf{execIdeal}(\Z, \F{}, \S_\A)$ are indistinguishable.
 
The simulation based security framework supports modular composition: we can build a protocol that emulates an intermediate functionality $\F{Hybrid}$ in the real world, and then build a high-level protocol that makes use of $\F{Hybrid}$ to realize the target functionality $\F{}$. The composition theorem guarantees that we can make this substitution.

\paragraph{SIDs} In UC, each functionality is associated with a unique string, called the session ID (SID). The SID is essential for the composition theorem, as it ensures that concurrent instances of protocols are kept separate from each other. The practical significance of the session ID is that it is implicitly used as a tag for signatures and hashes to ensure that messages from one protocol instance cannot be replayed in another.
To reduce clutter, we elide the handling of SIDs from our presentation.

\paragraph{Smart contracts and functionalities}
We define experiments with multiple ideal functionalities as well as ``contract'' processes, which represent programs running on the blockchain network.
Our ideal functionalities are round-based, refer to~\cite{TCC:KMTZ13,journals/joc/GarayMPY11,hawk} on how to modify the UC framework to support a synchronous network model with a computation that proceeds in rounds.
We treat the contracts as ideal functionalities too, which are available to protocols in our ``real world'' (i.e., the starting assumption for our work is that we have access to a blockchain primitive) following prior works using formal framework to capture the blokchain model~\cite{hawk,C:BK14,CCS:KMB15,EC:GKL14,EC:PSs17}.
The notion of multiple functionalities is compatible with UC --- that is, the functionalities can be considered as a single combined
functionality. For example, the inputs provided to two ideal functionalities ${\cal F}_1,{\cal F}_2$ at the start of the $k^\text{th}$ round will determine their joint output at the end of the $k^\text{th}$ round.

\paragraph{Delayed tasks}
Frequently in our ideal functionalities, we use the notation ``within \{R\} rounds: \{ \emph{Task} \}''. This is intended to guarantee that the pseudocode described by \emph{Task} is executed within a bounded time, but the exact time when it is executed is under the control of the adversary.
This mechanism is compatible with the traditional UC paradigm; we can imagine implementing a ``task queue'' mechanism within the functionality.
%

\paragraph{Exceptions in Ideal Functionalities}
To simplify our ideal functionality $\Fchain$, we allow the functionality to raise an exception. Raising an exception immediately sends an \mtt{Exception} message to the environment. Since in the real world there is no such mechanism for raising an exception,
\footnote{Assertions in smart contracts, such as Figure~\ref{fig:prot:statechannel} cause the activation to be discarded, without notifying the environment.}
this would clearly allow the environment to distinguish between the real and ideal worlds. Therefore in our security proof we have the obligation of showing that the simulator we construct never triggers an exception.

\else
\fi

\subsection{Simulation-based Security Proof for Payment Channels}
We now explain how to prove that the protocol $\Pi_{\msf{Pay}}$ realizes the ideal payment channel functionality $\F{Pay}$ in the $\F{State}$-hybrid world.
In the hybrid world, parties run the local protocol, exchanging point-to-point messages as well as interacting with the $\F{State}(U_\msf{Linked})$ functionality and the $\msf{Contract}_{\msf{Linked}}$ smart contract. In the ideal world, the parties simply communicate with $\F{Pay}$.
The security proof consists of a simulator $\S$ that translates every behavior in the hybrid world (Figure~\ref{fig:pictproof} Left) into an indistinguishable behavior in the ideal world (Figure~\ref{fig:pictproof} Right).

The simulator we construct is deterministic, and ensures that the hybrid world and ideal world are exactly identical in the view of the environment $\Z$. In Figure~\ref{fig:pictproof} we illustrate the case where party $P_\msf{L}$ is corrupted. We use color coding to denote the different types of messages. The simulator runs a copy of the hybrid world ``in its head,'' which it maintains in exact correspondence. Interactions between $\Z$ and corrupted $P_\L$ (red) in the hybrid world take the form of instructions to the dummy adversary to pass through to the hybrid world functionality $\F{State}$. Inputs to the honest party $P_\R$ correspond to inputs accepted by the ideal functionality, $\F{Pay}$, since these are passed through by the ideal protocol. Since $\F{Pay}$ leaks these immediately to the simulator, the simulator passes these on to the instance of $\Pi_{\msf{Pay}}$ it runs.

Per Cannetti~\cite{FOCS:Canetti01}, it suffices to construct a simulator $\S$ for the dummy adversary (i.e., the hybrid world adversary that simply follows instructions from the environment).
Since the model does not provide any secrecy, and since the $\F{State}$-hybrid world hides any cryptography, the simulation is information-theoretic and deterministic. The simulator runs a local sandboxed execution of $\Pi_\msf{Pay}$, which it keeps in perfect correspondence with the state of $\F{Pay}$. When the environment asks $\A$ to command corrupted parties to interact with and observe $\F{State}$, the simulator $\S$ routes these requests to its sandboxed $\Pi_\msf{Pay}$. We can show that the sandboxed execution of $\Pi_\msf{Pay}$ maintained by $\S$ is identical to $\Pi_\msf{Pay}$ in the hybrid world.

\begin{theorem}
The $\Pi_\msf{Pay}$ protocol in the $\F{State}$-hybrid world realizes the $\F{Pay}$ ideal functionality.
\end{theorem}
\begin{proof} (Sketch) Let us assume that $P_i$ is corrupt and $P_{\neg i}$ is honest. The ideal world simulator $\S$ for the dummy hybrid world adversary runs a sandboxed execution of $\Pi_\msf{Pay}$ through which it relays instructions as described below:

\textit{Inputs from honest parties.} When the simulator $\S$ receives a message of the form $(\mtt{pay},P_i,\$X)$ from $\F{Pay}$, it provides an input $\mtt{pay}(\$X)$ to $P_i$ in the sandboxed execution of $\Pi_\msf{Pay}$; when $\S$ receives $(\mtt{withdraw},P_i,\$X)$, it inputs $\mtt{withdraw}(\$X)$.

  \textit{Contract inputs.} If $\A$ instructs $P_i$ to send $\mtt{deposit}(\$X)$ to $\msf{Contract}_\msf{Pay}$, $\S$ simulates $P_i$ to send $\mtt{deposit}(\$X)$ to the sandboxed $\msf{Contract}_\msf{Pay}$, and translates the $O(\Delta)$-delayed task of $\F{State}$ to an $O(\Delta)$-delayed task for $\F{Pay}$.

  \textit{Message delivery in $\F{State}$.} When the environment asks to execute a delayed task in $\F{State}$ (i.e., to advance $\auxin$ or to apply a state update), $\S$ routes this request to $\Pi_\msf{Pay}$. If the sandboxed $\F{State}$ provides output $\msf{state}$ to a corrupted party, pass $\msf{state}$ to the environment.

  \textit{Outputs to honest parties.} If honest party $P_{\neg i}$ in the sandboxed $\F{State}$ provides output $(\msf{receive},\$X)$, then deliver the delayed task in $\F{Pay}$ that sends the same output to $P_{\neg i}$ in the ideal world.

\begin{figure*}
  \centering
  \includegraphics[width=5.5in]{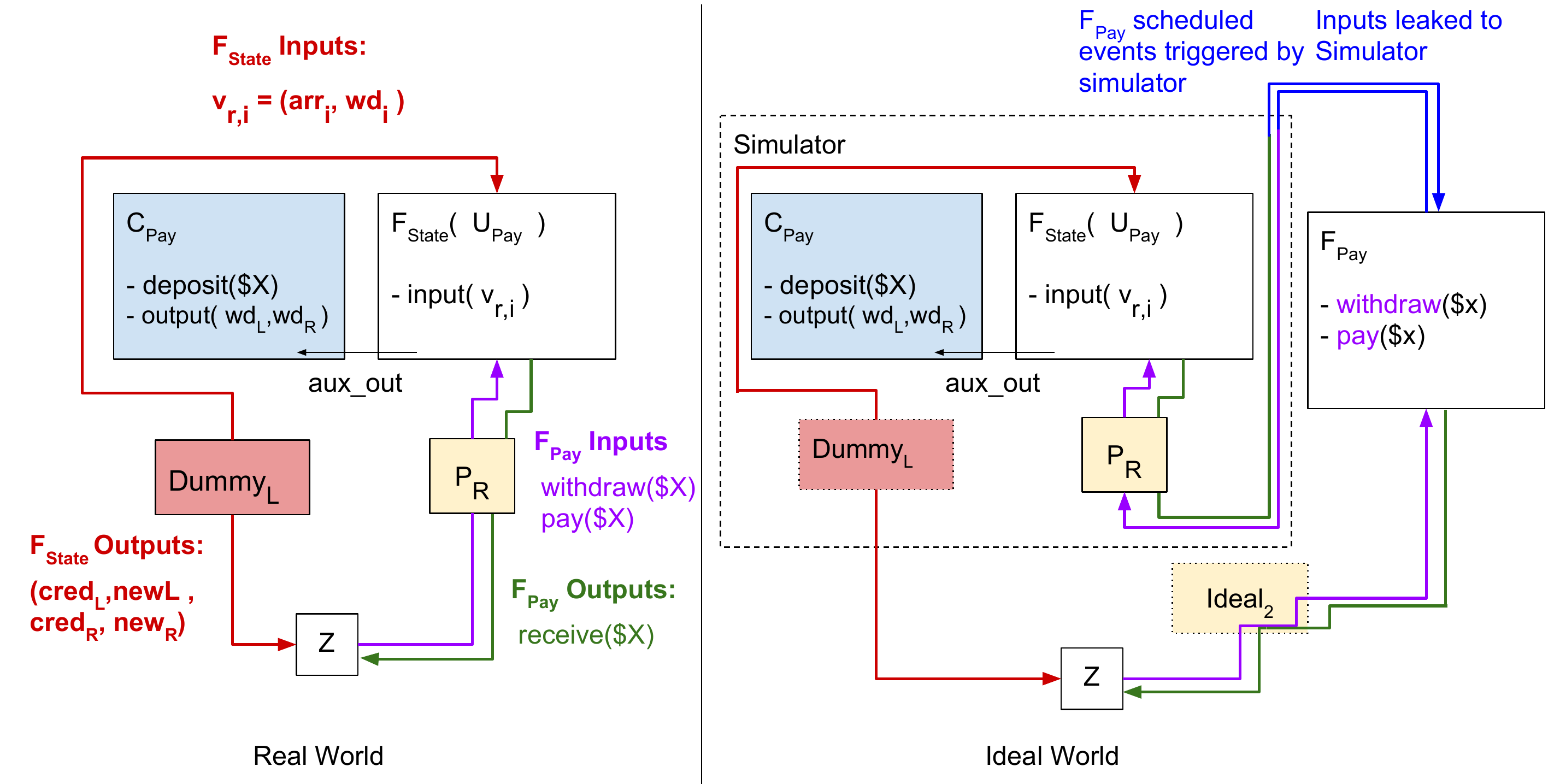}
  \caption{Illustration of the simulator-based proof for $\Pi_\msf{Pay}$ for corrupted $P_\L$ (the corrupted $P_\R$ case is symmetric.}
  \label{fig:pictproof}
\end{figure*}

  \smallskip
  The view that $\S$ constructs is identical to that of the hybrid world, because the credit of each $P_i$ in the hybrid world is kept in lockstep with $\msf{bal}_{i}$ of $\F{Pay}$ in the ideal world. Specifically, the ``while'' loop of $U_\msf{Pay}$ delivers each payment to the ideal $\F{State}$ by updating its \msf{state} in the same order that payments will be received by the ideal $\F{Pay}$. This allows $\S$ to schedule payments in the ideal world within $O(1)$ or $O(\Delta)$ rounds, so that $\Z$ will be able to observe each event (e.g., by inspecting the output of the honest $P_{\neg i}$) at exactly the same round in which the event occurred in the hybrid world. During each virtual round of $\F{State}$, the execution of $\Pi_\msf{Pay}$ by the honest $P_{\neg i}$ will thus correspond to the behavior of $P_{\neg i}$ in the ideal world, since the variables in the comparison $\$X \le \msf{Contract}_\msf{Pay}.\msf{deposits}_i+\msf{paid}_i{-\msf{pay}_i-\msf{wd}_i}$ are up-to-date and any new payments that $\S_\A$ scheduled would arrive only in the next virtual round.
  \end{proof}

\subsection{Details of the Linked Payments Construction }
In the body of the paper (Section~\ref{sec:statechannel}) we presented the update function and auxiliary smart contracts (Figure ~\ref{fig:prot-chain-basic}) for the state channel protocol $\Pi_\msf{Linked}$. In Figure~\ref{fig:prot-chain} we define the local behavior of the parties.
%

\iftr
\begin{figure*}
  \centering
        { \bf  Protocol $\Pichain(\$X, T, P_1, ... P_\ell)$ \\}
        \vspace{2pt}
        \input{sections/prot_chain_contract}%
\mbox{\vspace{1pt}\hspace{\columnsep}}%
\begin{minipage}[t]{\columnwidth}
  \else

\begin{figure}
  \centering
      { \bf  Protocol $\Pichain(\$X, T, P_1, ... P_\ell)$ \\}
      \vspace{2pt}
      \begin{minipage}[t]{\columnwidth}
\fi 
      
  \begin{framed}
    \small
{ \bf  Local protocol for sender, $P_1$ }
\vspace{-4pt}
\begin{itemize}
\item[] on \textbf{input} \mtt{pay} from the environment:
  \begin{itemize}
  \item[] $x \samples \{0,1\}^\lambda$, and $h \leftarrow \hash(x)$
  \item[] pass $(\mtt{open},{h},\$X, T_\msf{Expiry})$ as input to $\F{State}^1$
  \item[] send $(\mtt{preimage}, x)$ to $P_\ell$
  \item[] if $(\mtt{preimage}, x)$ is received from $P_2$ before $T_\msf{Expiry}$, then pass $\mtt{complete}$ to $\F{State}^1$
  \end{itemize}
\item[] \textbf{at time} $T_\msf{Expiry}+\Delta$, if $\msf{PM}.\msf{published}(T_\msf{Expiry}, h)$, then
  \begin{itemize}
  \item[] pass input $\mtt{complete}$ to $\F{State}^1$
  \end{itemize}
\item[] \textbf{at time} $T_\msf{Dispute}$, then pass input $\mtt{dispute}$ to $\F{State}^{1}$
\end{itemize}

\hrule
\vspace{4pt}

{\bf  Local protocol for party $P_i$, where $2 \le i \le \ell - 1$  }
\begin{itemize}
\item[] \textbf{on receiving state} $(\mtt{inflight}, h, \_)$ from $\F{State}^{i-1}$
  \begin{itemize}
  \item[] store $h$
  \item[] provide input $(\mtt{open}, h, \$X, T_\msf{Expiry})$ to $\F{State}^{i}$
  \end{itemize}
\item[] \textbf{on receiving state} $(\mtt{cancel}, \_, \_)$ from $\F{State}^{i}$,
  \begin{itemize}
  \item[] provide input $(\mtt{cancel})$ to $\F{State}^{i-1}$
  \end{itemize}
\item[] \textbf{on receiving} $(\mtt{preimage}, x)$ from $P_\ell$ before time $T_\msf{Crit}$, where $\hash(x) = h$,
  \begin{itemize}
  \item[] pass $\mtt{complete}$ to $\F{State}^{i}$
  \item[] \textbf{at time} $T_\msf{Crit}$, if state $(\mtt{complete}, \_, \_)$ has not been received from $\F{State}^{i}$, then
    \begin{itemize}
    \item[] pass contract input $\msf{PM}.\mtt{publish}(x)$
    \end{itemize}
  \end{itemize}
\item[] \textbf{at time} $T_\msf{Expiry}+\Delta$,
  \begin{itemize}
  \item[] if $\msf{PM}.\mtt{published}(T_\msf{Expiry}, h)$, pass $\mtt{complete}$ to $\F{State}^i$
  \item[] otherwise, pass $\mtt{cancel}$ to $\F{State}^{i-1}$
  \end{itemize}
\item[] \textbf{at time} $T_\msf{Dispute}$, pass input $\mtt{dispute}$ to $\F{State}^{i-1}$ and $\F{State}^{i}$
\end{itemize}

\hrule \vspace{4pt}
{ \bf  Local protocol for recipient, $P_\ell$ }
\begin{itemize}
\item[] \textbf{on receiving} $(\mtt{preimage},x)$ from $P_1$, store $x$ and $h := \hash(x)$~
\item[]   \vspace{2pt} \textbf{on receiving state} $(\mtt{inflight}, h, \_)$ from $\F{State}^{\ell-1}$,
  \begin{itemize}
  \item[] multicast $(\mtt{preimage},x)$ to each party
  \item[] \textbf{at time} $T_\msf{Crit}$, if state $(\mtt{complete}, \_, \_)$ has not been received from $\F{State}^{\ell}$, then
    \begin{itemize}
    \item[] pass contract input $\msf{PM}.\mtt{publish}(x)$
    \end{itemize}
  \end{itemize}
\item[] \textbf{at time} $T_\msf{Dispute}$, pass input $\mtt{dispute}$ to $\F{State}^{\ell-1}$
\end{itemize}

\hrule \vspace{5pt}
\textbf{\color{blue} Other messages.} Messages involving the $\color{blue} \F{Pay}$ interface are routed between the environment and the $\F{State}$ functionality according to $\color{blue} \Pi_\msf{Pay}$ (see {\color{blue} Figure~\ref{fig:protpay}})
  \end{framed}
\end{minipage}

\iftr
\caption{Construction for $\Fchain$ in the $\F{State}$-hybrid world. Portions of the update function $U_{\msf{Linked},\$X}$ that are delegated to the underlying $U_\msf{Pay}$ update function (Figure~\ref{fig:prot-chain-basic}) are colored {\color{blue}blue} to help readability. The left column is duplicated from Figure~\ref{fig:prot-chain-basic}.}
\label{fig:prot-chain}
\end{figure*}
\else 
\caption{Construction for $\Fchain$ in the $\F{State}$-hybrid world. (Local portion only. See Figure~\ref{fig:prot-chain-basic} for the smart contract portion.) Portions of the update function $U_{\msf{Linked},\$X}$ that are delegated to the underlying $U_\msf{Pay}$ update function (Figure~\ref{fig:prot-chain-basic}) are colored {\color{blue}blue} to help readability.}
\label{fig:prot-chain}
\end{figure}
\fi 

We now state and provide a proof sketch of our main
theorem for \F{Linked}.

\begin{theorem}
Protocol $\Pichain$ realizes \F{Linked}\
functionality in the \F{State}-hybrid world.
\end{theorem}
\begin{proof} (Sketch)
The ideal world simulator
$\S$ for the dummy real world adversary runs a
sandboxed execution of $\Pi_\msf{Linked}$ through
which it relays instructions while faithfully
simulating honest parties exactly as described in
$\Pichain$.
Note that in the simulation, \sim\ would act both
as \F{State}.
%
While \sim\ itself has quite a rich interface via
\F{Linked}, the crux of the proof will be in
proving that during the course of this
interaction, \F{Linked}\ never raises an
\mtt{Exception}.

\begin{proposition}
For each $i \in 1...(\ell-1)$, $\msf{flag}_i$ must
be in a terminal state, i.e., $\msf{flag}_i \in \{
\mtt{cancel}, \mtt{complete} \}$ at time $T +
O(\ell + \Delta)$ (or at time $T + O(\ell)$ when
all parties are honest). 
\label{prop:one}
\end{proposition}
\noindent{\it Proof sketch.} 
First we consider the case when all parties are
honest (and each party has sufficient balance). In
this case, the sender $P_1$ starts a
fresh round on its state channel (with $P_2$)
where it presents 
$h \leftarrow \mathcal{H}(x)$ for a randomly chosen
$x$. Note that $P_1$ also sends $x$ to $P_\ell$.
Observe that the \msf{flag}\ variable on $\F{State}^1$
state channel will be set to \msf{inflight},
following which (honest) party $P_2$ would start a
fresh round on its state channel with $P_3$ where
it simply forwards $h$ that it received from
$\F{State}^1$. This process repeats until it is
$P_{\ell}$'s turn where $P_\ell$ multicasts the
preimage $x$ that it received from $P_1$.
Then, each (honest) party $P_i$ simply passes a
\mtt{complete}\ command to $\F{State}^i$
Note that each of these steps take
one time step. Therefore, the whole protocol
completes within $\ell + 2$ time steps (all
off-chain), i.e., by 
time $T + O(\ell)$. Furthermore, each of the flag
variables $\msf{flag}_i$ would be in a terminal
state \mtt{complete}. 

When not all parties are honest, then it is
possible that $\F{State}^i$ might receive a
\mtt{cancel}\ command from corrupt $P_i$.
In this case, the \mtt{cancel}\ command is
propagated all the way back to $P_1$ by the honest
parties. In this case, all the local flag
variables in $\F{State}^i$ are set to
\mtt{cancel}\ and $P_1$ ends getting its deposit
back.
Note that the \mtt{cancel}\ event is confirmed on
the on-chain auxiliary contract. Since this is
settled on-chain, we incur an additional $\Delta$
delay.
One final scenario is when each $\F{State}^i$ is
set to \msf{inflight}\ (i.e., in the forward
path), and $P_\ell$ revealed the preimage but a
corrupt party does not agree to the update
off-chain. Once again, the state channel
synchronization process escalates to the on-chain
contract where this will be settled, and the flag
variables will be set in this case to
\mtt{complete}.
Note that the on-chain escalation will induce an
additional $\Delta$ delay (but this happens in
parallel for each state channel instance) and
therefore the proposition holds in this case as
well. 
This concludes the proof of the
proposition. 

\begin{proposition}
For each $i \in 1...(\ell-2)$, if $P_i$ is
honest, it must not be the case that
$(\msf{flag}_i, \msf{flag}_{i+1}) =
(\mtt{cancel},\mtt{complete})$  at time $T +
O(\ell + \Delta)$ (or at time $T + O(\ell)$ when
all parties are honest). 
\label{prop:two}
\end{proposition}
\noindent{\it Proof sketch.} 
We discuss the case when all parties are
honest. As described in the proof of
Proposition~\ref{prop:one}, when all parties are
honest, we have the payment will be settled within
time $T + O(\ell)$, and furthermore the flag
variables will be such that $\msf{flag}_i =
\mtt{complete}$ for all $i$.
Next we focus on the interesting case where some
parties are corrupt.
Then, we need to prove that the \msf{flag}\
variables are pairwise consistent.
When $P_i$ and $P_{i+1}$ are both honest, this
follows from the fact that state channels remain
synchronzied no matter how the payment between
$P_1$ and $P_\ell$ is settled. The interesting
case is when say $P_i$ is honest (condition in
proposition statement) and $P_{i+1}$ is corrupt.
That is, it suffices to prove that if
$\msf{flag}_i = \mtt{cancel}$, then it must hold
that $\msf{flag}_{i+1} = \mtt{cancel}$ also
holds.
Now if $P_\ell$ was corrupt and
$P_{i+1}$ revealed the hash preimage (received
from $P_\ell$) but the payment settlement
was escalated to the preimage manager contract
where the preimage was revealed, then it follows
from the protocol description that in this case
$\msf{flag}_i = \mtt{complete}$ (i.e., the
proposition precondition does not hold).
The remaining cases are relatively straightforward as
the only way $\msf{flag}_i$ is set to
\mtt{cancel}\ is when $P_{i+1}$ supplied
\mtt{cancel} to the state channel between $P_i$
and $P_{i+1}$ (and we already handled the case
when the preimage manager is involved in changing
$\msf{flag}_i$ to \mtt{complete}. 
This concludes the proof of the proposition.

\begin{proposition}
If $P_1$ and $P_\ell$ are honest, then
$(\msf{flag}_1,\msf{flag}_{\ell-1}) \in \{
(\mtt{complete},\allowbreak \mtt{complete}),
(\mtt{cancel},\allowbreak \mtt{cancel}) \}$  at time $T +
O(\ell + \Delta)$ (or at time $T + O(\ell)$ when
all parties are honest).
\label{prop:three}
\end{proposition}
\noindent{\it Proof sketch.} 
When all parties are honest, the proposition
directly follows from the proof of
Propositions~\ref{prop:one} and~\ref{prop:two}.
The interesting case is when some parties are
corrupt.
Note that when both $P_1$ and $P_\ell$ are honest,
it follows that if $P_\ell$ received
$\msf{inflight}$ from $\F{State}^{\ell-1}$ then it
would multicast the preimage of $h$ to all
parties.
We analyze two cases depending on whether $P_\ell$
multicasted the preimage or not. Note that since
$P_1$ and $P_\ell$ are honest the adversary does
not have knowledge (except with negligible
probability) of the preimage unless it was
revealed by $P_\ell$.
Suppose $P_\ell$ revealed the preimage of
$h$. Then, in this case, we will argue
that the end state for the flag variables will be
\mtt{complete}.
This is because once the preimage was revealed,
all honest parties become aware of it and each
honest $P_i$ would issue a \mtt{complete}\ command
to $\F{State}^i$ and either synchronize the state
on-chain (with the help of the preimage manager
contract) or off-chain.
This combined with the fact that $P_1$ is honest
implies that the end state for the flag variables
$\msf{flag}_1, \msf{flag}_\ell$ would be
\mtt{complete}.
We now turn to the case when $P_\ell$ never
revealed the preimage.
Then from the protocol description it follows that
the adversary is not aware of the preimage and
therefore cannot transition a state channel to
\mtt{complete}.
Therefore in this case, each state channel will
result in canceling the \msf{inflight}\ status,
and it follows that both $P_1$ and $P_\ell$ will
end up with flag variables $\msf{flag}_1,
\msf{flag}_\ell$ set to \mtt{cancel}.
This concludes the proof. 
\end{proof}

\subsection{Local Protocol for the State Channel Construction }
\label{sec:detail:statechannel}
In the body of the paper (Figure ~\ref{fig:prot:statechannel:contract}) we presented the smart contract portion of the state channel protocol. In Figure~\ref{fig:prot:statechannel} we define the local behavior of the parties.

\paragraph*{Reaching agreement off-chain}
The main role of the local portion of the protocol is to reach agreement on which inputs to process next. To facilitate this we have one party, $P_1$, act as the leader.
The leader receives inputs from each party, batches them, and then requests signatures from each party on the entire batch. After receiving all such signatures, the leader sends a $\mtt{COMMIT}$ message containing the signatures to each party.
This resembles the ``fast-path'' case of a fault
tolerant consensus protocol~\cite{pbft};
However, in our
setting, there is no need for a view-change
procedure to guarantee liveness when the leader
fails; instead the fall-back option is to use the
on-chain smart contract.

\iftr
\begin{figure*}[!ht]
        { \centering \bf  Protocol $\Pi_\msf{State}(U, P_1, ... P_N)$ \\}
        \vspace{2pt}

\input{sections/prot_state_contract}%
\mbox{\vspace{1pt}\hspace{\columnsep}}%
\begin{minipage}[t]{\columnwidth}
\else 
\begin{figure}[!ht]
  { \centering \bf  Protocol $\Pi_\msf{State}(U, P_1, ... P_N)$ \\}
  \vspace{2pt}
  \begin{minipage}[t]{\columnwidth}
\fi 
          
  \begin{framed}
    \small
      {\centering \bf  Local protocol for the leader, $P_1$ \\}
    
  \begin{itemize}
  \item[] Proceed in consecutive virtual rounds numbered $r$:
    \begin{itemize}
    \item[] Wait to receive messages $\{ \mtt{INPUT}( \vrj) ) \}_j$ from each party.
    \item[] Let $\inr$ be the current state of $\auxin$ field in the the contract.
    \item[] Multicast $\mtt{BATCH}( r, \inr, \{ \vrj \}_j )$ to each party.
    \item[] Wait to receive messages $\{ (\mtt{SIGN}, \srj ) \}_j$ from each party.
    \item[] Multicast $\mtt{COMMIT}(r, \{ \srj \}_j )$ to each party.
    \end{itemize}
    \end{itemize}

  \hrule
  \vspace{5pt}
    
    {\centering {\bf  Local protocol for each party $P_i$} (including the leader) \\}
  \begin{itemize}
  \item[] $\msf{flag} := \mtt{OK} \in \{ \mtt{OK},
  \mtt{PENDING} \}$
  \item[] Initialize $\msf{lastRound} := -1$
  \item[] Initialize $\msf{lastCommit} := \bot$
  \end{itemize}
{\bf Fast Path} (while $\msf{flag} == \mtt{OK}$) \\
\begin{itemize}
\item[] Proceed in sequential rounds $r$, beginning with $r := 0$
\item[] Wait to receive input $\vri$ from the environment. Send $\mtt{INPUT}(\vri)$ to the leader.
\item[] Wait to receive a batch of proposed
inputs, $\mtt{BATCH}(r,\allowbreak
\inrp,\allowbreak \{ \vrjp \}_j )$ from the
leader. Discard this proposal if $P_i$'s own input
is omitted, i.e., $\vrip \neq \vri$. Discard this proposal if $\inrp$ is not a \emph{recent} value of $\auxin$ in the contract.
\item[] Set $(\msf{state},\msf{out}_r) := U( \msf{state}, \{\vrj\}_j, \inrp )$
\item[] Send $(\mtt{SIGN},\sri)$ to $P_1$, where $\sri := \msf{sign}_i(r \| \msf{out}_r \| \msf{state})$
\item[] Wait to receive $\mtt{COMMIT}(r,\{\srj \}_j )$ from the leader. Discard unless $\msf{verify}_j(\srj \| \msf{out}_r \| \msf{state})$ holds for each $j$. Then:
  \begin{itemize}
  \item[] $\msf{lastCommit} := (\msf{state},
  \msf{out}_r, \{\srj \}_j)$; $\msf{lastRound} :=
  r$ 
  \item[] If $\msf{out}_r \neq \bot$, invoke $\mtt{evidence}(r, \msf{lastCommit})$ on the contract.
  \end{itemize}
  \item[] If  $\mtt{COMMIT}$ was not received
  within one time-step, then:
  \begin{itemize}
  \item[] if $\msf{lastCommit} \neq \bot$, invoke the
  $\mtt{evidence}(r-1, \msf{lastCommit})$ and
  $\mtt{dispute}(r)$ methods of $C$
  \end{itemize}
\end{itemize}

{\centering \bf Handling on-chain events}

\begin{itemize}
\item[] On receiving $\mtt{EventDispute(r,\_)}$, if $r \le \msf{lastRound}$, then
invoke $\mtt{evidence}(\msf{lastRound}, \msf{lastCommit})$ on the contract. Otherwise if $r = \msf{lastRound} + 1$, then:
  \begin{itemize}
  \item[] Set $\msf{flag} := \mtt{PENDING}$, and
  interrupt any ``waiting'' tasks in the fast path
  above. Inputs
  are buffered until returning to the fast path.
\item[] Send $\mtt{input}(r, \vri)$ to the contract.
\item[] Wait to receive $\mtt{EventOffchain}(r)$
  or $\mtt{EventOnchain}(r)$ from the contract.
  Attempt to invoke $\mtt{resolve}(r)$ if $\Delta$ elapses, then continue waiting. In
  either case:
  \begin{itemize}
    \item[] $\msf{state} := \msf{state'}$
    \item[] $\msf{flag} := \mtt{OK}$
    \item[] Enter the fast path with $r := r + 1$
    \end{itemize}
  \end{itemize}
\end{itemize}

  \end{framed}
\end{minipage}

\iftr
\caption{Construction of a general purpose state channel, $\F{State}$, parameterized by transition function $U$. For readability, the left column is duplicated from Figure~\ref{fig:prot:statechannel:contract}.}
  \label{fig:prot:statechannel}
\end{figure*}
\else 
\caption{Construction of a general purpose state channel, $\F{State}$, parameterized by transition function $U$. (Local portion only, for the smart contract see Figure~\ref{fig:prot:statechannel:contract}.)}
\label{fig:prot:statechannel}
\end{figure}
\fi 

 \label{sec:prot:state:appendix}

\iftr

\paragraph{Invariants maintained by the state channel}
Note that that when $\mtt{EventDispute}(r)$ is received, we can assume that $r \le \msf{bestRound}+1$ in the view of any honest party. This is because the $\mtt{EventDispute}$ event can only be triggered for round $\msf{bestRound}+1$ of the contract, and $\msf{bestRound}$ is only set in the contract after receiving inputs from every honest party.

If a party $P_i$ receives $\mtt{EventOffchain}(r)$ while in the $\msf{flag} = \mtt{DISPUTE}$ condition, it can be safely assumed that a $\mtt{BATCH}$ message has already been received and $\msf{state}$ has already been updated to reflect the new state. This is because $\mtt{EventOffChain}$ can only be triggered after receiving an update containing signatures from all parties, including $P_i$.

\begin{lemma}
If any party receives $\msf{EventDispute}(r, T)$,
then every party will receive from the contract
either (1)
$\msf{EventOffchain}(r)$ within time $T + \Delta$,
or (2) $\msf{EventOnchain}(r)$ within time $T +
2\Delta$.
\label{lem:completion}
\label{lem:either}
\end{lemma}
\begin{proof} (Sketch)
Since some party received $\msf{EventDispute}$
from the contract, this means that the contract
received $\mtt{dispute}$ from one of the
parties.
Observe that the variable $\msf{deadline}$ is set
to $T + \deltareceive$.
Now, if $r$ corresponds to an earlier round (i.e.,
not the current incomplete 
round), then an honest party upon receiving
$\msf{EventDispute}$ would send an $\mtt{update}$
command to the contract with the state (along with
signatures) corresponding to the most recent
completed round.
In this case, it follows that the contract would
emit $\msf{EventOffChain}$ before time $T +
\deltareceive$.
On the other hand, suppose $r$ corresponds to
the current round.
In this case, we have that honest parties would
have $r = \msf{lastRound}+1$, and they would send
their current round input $(r,\vri)$ to the
contract (which would be accumulated by $C$) as
their response to $\msf{EventDispute}$.
Then, at time $T+\Delta$, honest parties would
send $\mtt{resolve}$ right after $T+\Delta$
which ensures that they get a response from the
on-chain contract before time $T + 2\Delta$.
This concludes the proof sketch. 
\end{proof}

\begin{lemma}
If any honest party receives
$\msf{EventDispute}(r,T)$ from the contract, but
does not receive $\mtt{COMMIT}(r,\_)$ from the
leader before time $T$, then no party will receive
$\mtt{COMMIT}(r+1,\_)$ until receiving
$\msf{EventOnchain}(r)$ or
$\msf{EventOffchain}(r)$ from the contract.
\label{lem:nocom}
\end{lemma}

\begin{proof} (Sketch)
Note that for the leader $P_1$ to generate a valid
$\mtt{COMMIT}$ message for round $r+1$, it needs
signatures from all parties.
Suppose some honest party received
$\msf{EventDispute}(r,T)$, then it must hold that
$r = \msf{bestRound}+1$.
This in particular means that the contract
received signatures on the state corresponding to
the $(r-1)$-th round from all parties. 
Since $P_1$ would not be able to forge honest
parties' signatures, it follows that the honest
parties must have synchronized their state until
the  $(r-1)$-th round (either on-chain or
off-chain).
Now, suppose it holds that some honest party $P_j$
did not receive a message $\mtt{COMMIT}$ message
for round $r$, then there are two cases to
handle.
First, if $r$ is not the current round, then in
this case honest parties $P_{j}$ would directly send an
$\mtt{update}$ command with the most recent
completed round along with the latest
synchronized state to the contract.
This has the effect of making the contract emit
$\msf{EventOffChain}$ for round $r$.
Next if $r$ is indeed the current round, then it
follows 
that honest party would
immediately send to the contract (1) an
$\mtt{update}$ message with the state
corresponding to the $(r-1)$-th round, and
immediately afterwards (2) a $\mtt{dispute}$
message for round $r$.
Taken together, these messages have the effect of
ensuring that the contract emits
$\msf{EventDispute}$ for round $r$.
Following this and since
$r = \msf{lastRound} + 1$, an honest party $P_i$
would respond by sending their inputs $\vri$ to
the contract (irrespective of whether honest $P_i$
received a message of the form
$\mtt{COMMIT}(r,\_)$).
Then, the rest of the state is synchronized
on-chain (i.e., parties submit their inputs
directly to the contract).
At time $T + \Delta$, honest parties would send
$\mtt{resolve}$ to the contract which would
result in the contract emitting
$\mtt{EventOnChain}$ for round $r$. This concludes
the proof sketch.
We defer this proof to the full version~\cite{fullonline}.
\end{proof}

\begin{lemma} If any honest party receives
$\msf{EventDispute}(r,T)$ from the contract, and
receives $\mtt{COMMIT}(r)$ from the leader before
time $T$, then every party will receive
$\msf{EventOffchain}(r)$ by time
$T+2\Delta$.
\label{lem:onchain}
\end{lemma}
\begin{proof} (Sketch).
By Lemma~\ref{lem:either}, it follows that we only need to show
that $\msf{EventOnchain}(r)$ will not be emitted
whenever $\mtt{COMMIT}(r)$ is received before time
$T$.
Given that some honest party $P_j$ received
$\mtt{COMMIT}(r)$ from the leader, then this means
that all parties must have synchronized their
state until round $r-1$.
In addition, since $P_j$ also received
$\msf{EventDispute}(r,T)$ from the contract, this
implies that some party issued $\mtt{dispute}$
to the contract. Observe that $r = \msf{lastRound} 
\msf{bestRound}+1$ also holds.
Now each honest party $P_i$ would respond to
the $\msf{EventDispute}$ message with an
$\mtt{update}$ command that would result in the
on-chain contract emitting
$\msf{EventOffchain}(r)$ within time $T +
\Delta$. This concludes the proof.
\end{proof}

\begin{theorem}
The $\Pi_{\msf{State}}$ protocol realizes the 
$\F{State}$ functionality assuming one way functions exist. 
\end{theorem}
\begin{proof} (Sketch) The ideal world simulator
$\S$ for the dummy real world adversary runs a
sandboxed execution of $\Pi_\msf{State}$ through
which it relays instructions as described below.

When the
simulator $\S$ receives a message of the form
$(i,\vri)$ from $\F{State}$ (i.e., leaked honest
inputs), it stores this message to use later in
the simulation for the current round.

Following this, \sim\ simulates the local protocol
for each honest player to the \adv\ using leaked
inputs.
%
If the leader is honest, then \sim\ acting as the
leader, uses the leaked inputs (from
$\F{State}$) as inputs received from honest
parties.
\sim\ waits to receive round inputs from \adv\ for
the corrupt parties.
If it receives inputs from all corrupt parties,
then acting as the leader, \sim\ accumulates these
values, and multicasts it to all corrupt parties.
Following this, \sim\ waits to receive signatures
on the batched inputs. \sim\ generates signatures
on behalf of the honest parties (i.e., it uses a
simulated public/signing keys for the honest
parties). Once all signatures are received, then
these signatures (i.e., corresponding to both
honest and corrupt parties) are multicasted to the
adversary. \sim\ then increments the round number
and continues with the simulation.
On the other hand, if the leader is corrupt, then
\sim\ sends simulated honest messages (according
to $\Pi_\msf{State}$ to \adv\ and
receives messages \mtt{BATCH}\ and \mtt{COMMIT}\
from \adv. 

Of course, all this is good only when the corrupt
parties follow the off-chain protocol faithfully.
Now a corrupt party might not send its inputs or
its signature or might directly trigger the
contract or might attempt to keep the contract and
off-chain executions out of sync.
Not supplying inputs is handled in a
straightforward manner by just replacing it with
$\bot$.
Not supplying signatures on batched inputs
corresponds to a halt in the off-chain round.
This in turn would result in honest parties having
to escalate the off-chain round to the on-chain
contract. Lemma~\ref{lem:either} then shows how
the execution gets completed on-chain. Since \sim\
closely mimics the real protocol, the
indistinguishability follows from the proof of
Lemma~\ref{lem:either} in this case.
Directly raising a dispute the on-chain contract would
result in a state change that notifies the
(simulated) honest parties who can then issue an
\mtt{update}\ command to the contract. This case
is handled by Lemma~\ref{lem:onchain}.
Finally, \adv's attempts to keep the on-chain
contract and off-chain executions out of sync are
handled by observing that on-chain contract
notifications (in particular \msf{EventDispute})
are available to all honest parties who can then
recover from inconsistent corrupt inputs by
issuing an \mtt{update}\ command to the contract
that syncs the state on the on-chain chain with
the off-chain state. This case is handled by
Lemma~\ref{lem:either} and
Lemma~\ref{lem:onchain}.

Finally, we consider the case 
when the leader is unable to provide a 
\mtt{COMMIT}\ message for this round. This could
happen if either the leader is corrupt, or the
corrupt parties did not submit a signature on the
\mtt{BATCH}\ message the contains this round's
messages.
In either case, Lemma~\ref{lem:nocom} applies and
we are guaranteed that the next round honest
messages are exchanged only after the current
round state is synchronized among the parties.

Contract inputs and outputs (including coins) are
handled in the straightforward manner. One thing
to note is that the event emitted by the on-chain
contract may be delayed. Other than this,
simulation proceeds by faithfully applying the
parameterized update function to the state and
also maintains a variable $\msf{bestRound}$ which
in addition to variables $\msf{state}$ and
$\msf{applied}$, keeps track of the on-chain state
(and prevents replay attacks).
\end{proof}
%

%


\fi

\iftr %

\subsection{Proof of concept Implementation in Ethereum}
\begin{table}[t]
\centering
\begin{tabular}{l@{\qquad}r@{\qquad}r}
Transaction & Cost in Gas & Cost in \$ \\ \hline 
Create Preimage Manager & $150,920$ & $0.27$ \\ 
Submit Preimage & $43,316$ & $0.08$ \\  
Create Sprites & $4,032,248$ & $7.26$ \\ 
Deposit & $42,705$ & $0.08$ \\
Store State (Hash) & $98,449$ & $0.18$ \\ 
Payment & $163,389$ & $0.29$ \\ 
Open Transfer & $323,352$ & $0.58$ \\
Complete Transfer & $62,235$ & $0.11$ \\  
Dispute Transfer & $63,641$ & $0.11$ \\ 
Dispute & $88,047$ & $0.16$ \\ 
Individual Input & $49,444$ & $0.09$ \\ 
Resolve & $199,151$ & $0.36$ \\ 
\end{tabular}
\caption{\label{fig:financialcost} A breakdown of the costs on Ethereum's test network. We approximate the cost in USD (\$) using the conversion rate of 1 ether = \$90 and the gas price of 0.00000002 ether which are the real world costs in May 2017. }
\end{table}

In this section, we discuss our proof of concept implementation of Sprites before highlighting how the implementation supports concurrent conditional tranfers and how to incorporate a payment protocol to fairly exchange an invoice for payment among merchants and customers.  

\paragraph{Implementation of the Smart Contracts} 

Table \ref{fig:financialcost} provides a break down of the computational and financial costs for deploying our implementation of Sprites on Ethereum's Ropsten test network.
There is a contract for the Preimage Manager\footnote{Manager \texttt{0x62E2D8cfE64a28584390B58C4aaF71b29D31F087}} and the Sprites channel\footnote{Sprites: \texttt{0x85DF43619C04d2eFFD7e14AF643aef119E7c8414}}.
Both contracts are written in Solidity and reflect the functionalities outlined in this paper.
Next, we briefly discuss the functionality of each contract and their associated costs.

\textbf{Preimage Manager.} This is a timestamping service that records when the pre-image of a conditional transfer is revealed in the blockchain. 
Creating this contract requires 150,920 gas and submitting the pre-image of an in-flight conditional transfer requires 43,316 gas. 
The address of this contract is included in the Sprites contract to facilitate communication for resolving disputed conditional transfers. 

\textbf{Sprites}. This contract establishes a two-party channel and replicates the functionalities outlined in the paper. 
Creating the Sprites contract costs 4,032,248 gas.
It is worth mentioning that Sprites can be implemented such that its code is re-usable by other contracts.
This is feasible using the \texttt{DELEGATECALL} opcode which permits a new contract to call the desired function in Sprites before updating the calling contract's local storage.
An example of this approach can be found here \cite{upgradeable}. This is important to highlight as the expensive gas cost for creating a Sprites contract is only incurred once.

Each party can deposit coins into the contract at any time and this costs 42,705 gas. 
In each round both parties co-operatively sign the new state's hash and this requires no interaction with the contract.
If there is a dispute, either party can broadcast a transaction that includes the current state hash to the Sprites contract. 
This executes the update function which will verify the counterparty's signature before recording the state hash and this costs 98,449 gas.
It is worth noting that the contract only accepts the state hash if it is more recent than what is currently recorded.  
Once the hash is stored in the contract, both parties must provide Sprites with the pre-image in order to update the contract to reflect this state. 

Unfortunately solidity restricts the number of variables that can exist in a single function \cite{FC:MCCORRY17} and this required us to create two seperate functions \texttt{UpdateChannel} and \texttt{UpdateTransfers} for updating the contract's state. 
The first function \texttt{UpdateChannel} accepts both the payment and withdrawal commands.
The former updates the balance of both parties to reflect a new payment, whereas the former deallocates coins for use in this channel such that the coins can later be withdrawn from the contract.
In our experiment we tested the payment command and this costs 163,389 gas. 
On the other hand, the second function \texttt{UpdateTransfers} applies the conditional transfers. 
Opening a transfer costs 323,352 gas and completing a successful transfer costs 62,235 gas, whereas raising a dispute and communicating with the PreimageManager costs 63,641 gas.  

Once the latest state is applied, one or both parties can send commands directly to the contract using the trigger functionality.  
Calling the \texttt{dispute} function will employ a grace period that permits one or both parties to submit commands to the contract and this costs 87,850 gas. 
As an example, we provided the payment command as an individual input during this grace period which cost 49,444 gas.  
To finish either party can call the \texttt{resolve} function after the grace period.
This effectively processes the commands (i.e. payment) before updating the contract's state for the next round. 
Of course, this trigger can be cancelled during the grace period if either party submits a state hash for a more recent round that is what currently stored in the contract.


Our implementation of the state channel contract is given in Figure~\ref{fig:prot:statechannel:solidity}.
\definecolor{mymauve}{rgb}{0.58,0,0.82}
\definecolor{mygreen}{rgb}{0,0.55,0}

\lstset{
  numbers=left,
  stepnumber=1,
  firstnumber=1,
  numberfirstline=true,
  language=java,
  morekeywords={modifier,function,contract,address,mapping,uint,payable,assert,verifySignature,sha3,event},
  keywordstyle=\bfseries\color{blue},
  stringstyle=\color{mymauve},
  commentstyle=\color{mygreen},
  escapeinside={(*@}{@*)},
  basicstyle=\ttfamily,
  }

\begin{figure*}
  \begin{minipage}{\textwidth}
    \scriptsize
    \begin{framed}
\begin{lstlisting}
contract StateChannel {
    address[] public players;
    mapping (address => uint) playermap;
    int bestRound = -1;
    enum Flag { OK, DISPUTE }
    Flag flag;
    uint deadline;
    mapping ( uint => bytes32[] ) inputs;
    mapping ( uint => bool ) applied;

    bytes32 aux_in;
    bytes32 state;

    event EventDispute (uint round, uint deadline);
    event EventOnchain (uint round);
    event EventOffchain (uint round);

    function handleOutputs(bytes32 state) {
	// APPLICATION SPECIFIC REACTION
    }

    function applyUpdate(bytes32 state, bytes32 aux_in, bytes32[] inputs) returns(bytes32) {
	// APPLICATION SPECIFIC UPDATE
    }

    function input(uint r, bytes32 input) onlyplayers {
	uint i = playermap[msg.sender];
	assert(inputs[r][i] == 0);
	inputs[r][i] = input;
    }
    
    function dispute(uint r) onlyplayers {
    	assert( r == uint(bestRound + 1) ); // Requires the previous state to be registered
	assert( flag == Flag.OK );
	flag = Flag.DISPUTE;
	deadline = block.number + 10; // Set the deadline for collecting inputs or updates
        EventDispute(r, block.number);
    }

    function resolve(uint r) {
	// No one has provided an "update" message in time
	assert( r == uint(bestRound + 1) );
	assert( flag == Flag.DISPUTE );
	assert( block.number > deadline );

	// Process all inputs received during trigger (a default input is used if it is not received)
	flag = Flag.OK;
	state = applyUpdate(state, aux_in, inputs[r]);
	EventOnchain(r);
	bestRound = int(r);
    }

    function evidence(Signature[] sigs, int r, bytes32 _state) onlyplayers {
        if (r <= bestRound) return;
	
        // Check the signature of all parties
        var _h = sha3(r, state);
        for (uint i = 0; i < players.length; i++) {
	    verifySignature(players[i], _h, sig);
        }

	// Only update to states with larger round number
	if ( r > bestRound) {
	    // Updates for a later round supercede any pending dispute
	    if (status == Status.DISPUTE) {
		status = Status.OK;
		EventOffchain(uint(bestRound+1));
	    }
	    bestRound = r;
            state = _state;
	}
	applied[r] = true;
	handleOutputs(_state);
    }
}        
\end{lstlisting}
    \end{framed}
  \end{minipage}
  \caption{Solidity contract for general purpose state channel, corresponding to the pseudocode in Figure~\ref{fig:prot:statechannel:contract}.}
  \label{fig:prot:statechannel:solidity}
\end{figure*}



\fi 

\subsection{Reproduction of the Flare experiment~\cite{flare}}
\begin{figure}[t]
  \centering
  \includegraphics[width=3.3in]{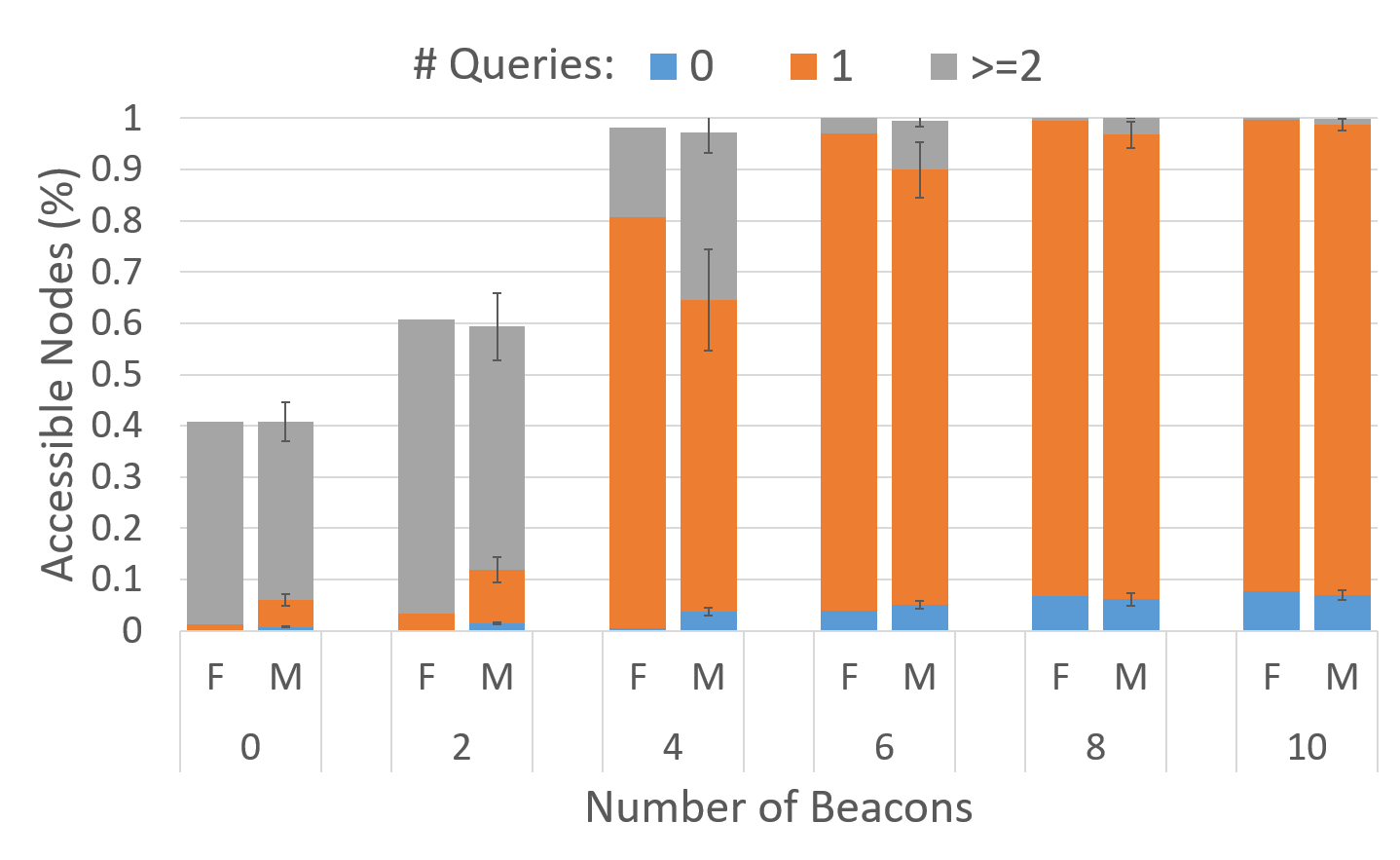}
  \caption{Payment channel routing success in Flare. The left columns (F) are from Pridhoko et al.~\cite{flare}, while the right columns (M) are from our implementation. If the accessible nodes percentage is $y$, then on average, a node in the network can route a payment to $y$ percent of the other nodes. We vary the Flare algorithm parameters, the number of peers each node queries, and the number of beacons kept in each node's routing table. }
  \label{fig:flare}
\end{figure}

Flare is one of the first proposed decentralized routing algorithms for payment channel
networks. In Flare, each node has a routing table consisting of the neighborhood of nodes that
are nearby in hop distance (connected by a path of few channels). Routing tables also contain
paths to a number of beacons, which serve as landmarks and allow the node to have a partial
view of the far away parts of the network. If a source node is unable to route to a destination
using its routing table, it continuous to query additional nodes for their routing tables until it is
either to determine a route or it gives up. Flare allows for a high expectation of finding a route
while requiring that each node to maintain memory proportionate to the logarithm of the total
number of nodes on the network.

\paragraph{Routing Table Generation}
In Flare, each node has a routing table consisting of the neighborhood of nodes that are nearby
in hop distance. Nodes also select a number of beacons, which serve as landmarks and allow
the node to have a partial view of the far away parts of the network. Like Kademlia, all nodes
are given the output of a hash function as an ID, so that the XOR of IDs can be used to measure
distance. Nodes attempt to find beacons that are closest to them in XOR address space.
Because these addresses are randomly generated, beacon nodes are expected to be average
distance away in hop length. During beacon selection, nodes close in XOR address space are
selected as beacon candidates. Beacon candidates can accept their role or give the path to an
even closer candidate that it knows of. This process continues until no additional beacon
candidates are suggested.

\paragraph{Route Selection}
During route selection, the source node first attempts to route to the destination using its
routing table. Failing that, it attempts to route using the combined routing table of both itself
and the destination node. Each failure after that, the source node not-yet queries the node in
its routing table that has the smallest XOR distance to the destination. The queried node
responds with its routing table, which the source node merges with its combined routing table.
The source node may continue to query nodes until it is able to route to the destination or give
up.

\paragraph{Beacon Selection}
Flare normally performs $k$-shortest paths using the combined knowledge of
the source, destination, and the routing tables of any queried nodes. Flare uses k-shortestpaths
to account for the dynamic nature of payment channel networks, which means channels
may not have enough capacity by the time route selection is made. Upon each failure, the
source queries an additional node for its routing table, which is then used to reattempt route
selection.
We approximate this behavior by performing a single shortest path over only the channels
known to have enough capacity to field the payment. Upon failing 10 times (and querying 10
additional nodes' routing tables), we mark the request as failed. Generally, payments may have to route
along longer paths or fail as channels close and capacity diminishes. Therefore, we expect that
our Flare approximation will exhibit signs of stress in the network by more frequently returning
longer path lengths or failing to return a path at all. 

Pridhoko et al.~\cite{flare} evaluated Flare with simulations of 2,000 and
100,000 nodes. In their simulations, they use the Watts-Strogatz topology with an average
degree of 4 and edge rewiring probability of 0.3. They parameterize their algorithm with a
neighborhood radius of 2 and a route selection query limit of 10. They vary the number of
beacons between 0 and 12.
Their simulation proceeds follows:
\begin{enumerate}
\item Initialize payment channel network topology.
\item Perform beacon selection for all nodes in random order.
\item Choose 10 nodes randomly. For each selected node attempt to route to all other
  nodes.
\end{enumerate}
We repeated their simulation 30 times using our implementation of their algorithm.
Comparisons of our results for accessible nodes are shown in Figure~\ref{fig:flare}. Our results also
feature error bars that are equal to 2 standard deviations both above and below the mean
value. Though there are some differences in our results, we notice that the differences
decrease as the number of queries and beacons increase. While some of the variation can be
explained by nondeterminism, we do expect that there are minor differences in
implementation. Regardless, because we use at least 6 beacons and a query limit of 10 in our
experimentation (where there is little difference in our implementations), we expect our
implementation to faithfully represent the proposed algorithm.


\end{document}